\documentclass[12pt,a4j]{article}
\usepackage{amsmath,amsthm,amssymb,bm,hhline,psfrag,comment}
\usepackage{url}
\usepackage{color}
\usepackage[dvipdfm]{graphicx}
\usepackage[numbers]{natbib}

\def\hsymbu#1{\smash{\lower1.7ex\hbox{\huge$#1$}}}

\usepackage{indentfirst}
\usepackage{algorithm,algorithmic}

\newcommand{\bbeta}{\bm{\beta}}

 \newcommand{\mcb}{\mathcal{B}} 
 
\newcommand{\mcf}{\mathcal{F}} \newcommand{\mci}{\mathcal{I}} \newcommand{\mcl}{\mathcal{L}} 
 
  \newcommand{\mcx}{\mathcal{X}}

\newcommand{\mbbn}{\mathbb{N}}
\newcommand{\mbbr}{\mathbb{R}}

\newcommand{\al}{\alpha} \newcommand{\lam}{\lambda}  
 
\newcommand{\ve}{\varepsilon} 
\newcommand{\del}{\delta}
\newcommand{\sig}{\sigma}
\newcommand{\D}{\Delta} 
\newcommand{\gam}{\gamma}

\newcommand{\p}{\partial}
 
\newcommand{\cil}{\xrightarrow{\mcl}} 
\newcommand{\cip}{\xrightarrow{p}} 

\newcommand{\ucip}{\stackrel{p}{\rightrightarrows}} 

\newcommand{\argmin}{\mathop{\rm argmin}}

\newcommand{\var}{{\rm var}}
\newcommand{\cov}{{\rm cov}}
\newcommand{\sumi}{\sum_{i=1}^{n}}

\newcommand{\bhg}{\hat{\bm{\beta}}_{\gam}}

\newcommand{\nn}{\nonumber}

\makeatletter
\def\eqnarray{\stepcounter {equation}\let \@currentlabel =\theequation
\global \@eqnswtrue
\global \@eqcnt \z@ \tabskip \@centering \let \\=\@eqncr
$$\halign to \displaywidth \bgroup \@eqnsel \hskip \@centering
$\displaystyle \tabskip \z@ {##}$&\global \@eqcnt \@ne \hfil
${\mbox{}##\mbox{}}$\hfil &\global \@eqcnt \tw@
$\displaystyle \tabskip \z@ {##}$\hfil \tabskip \@centering
&\llap {##}\tabskip \z@ \cr}
\makeatother

\theoremstyle{plain}
\newtheorem{lemma}{Lemma}[section]
\theoremstyle{remark}
\newtheorem{remark}{Remark}[section]
\theoremstyle{example}

\theoremstyle{lemma}

\newtheorem{Assumption}{Assumption}
\theoremstyle{Theorem}
\newtheorem{Theorem}{Theorem}[section]

\begin{document}

{
\begin{center}
\textbf{\Large Robust relative error estimation}
\end{center}
\begin{center}
\large {Kei Hirose $^{1,3}$ and Hiroki Masuda $^{2}$ 
}
\end{center}

\begin{flushleft}
{\footnotesize
$^1$ Institute of Mathematics for Industry, Kyushu University, 744 Motooka, Nishi-ku, Fukuoka 819-0395, Japan \\


$^2$ Faculty of Mathematics, Kyushu University, 744 Motooka, Nishi-ku, Fukuoka 819-0395, Japan \\

\vspace{1.2mm}

$^3$ RIKEN Center for Advanced Intelligence Project, 1-4-1 Nihonbashi, Chuo-ku, Tokyo 103-0027, Japan \\
}
{\it {\small E-mail: hirose@imi.kyushu-u.ac.jp, hiroki@math.kyushu-u.ac.jp}}	
\end{flushleft}

\vspace{1.5mm}

\begin{abstract}
Relative error estimation has been recently used in regression analysis.  A crucial issue of the existing relative  error estimation procedures is that they are sensitive to outliers.  To address this issue, we employ the $\gamma$-likelihood function, which is constructed through $\gamma$-cross entropy with keeping the original statistical model in use.   The estimating equation has a redescending property,  a desirable property in robust statistics, for a broad class of noise distributions.  To find a minimizer of the negative $\gamma$-likelihood function, a majorize-minimization (MM) algorithm is constructed.  The proposed algorithm is guaranteed to decrease the negative $\gamma$-likelihood function at each iteration.
We also derive asymptotic normality of the corresponding estimator together with a simple consistent estimator of the asymptotic covariance matrix, so that we can readily construct approximate confidence sets.
Monte Carlo simulation is conducted to investigate the effectiveness of the proposed procedure.  Real data analysis illustrates the usefulness of our proposed procedure.
 \end{abstract}
 \noindent {\bf Key Words}: $\gamma$-divergence; relative error estimation; robust estimation

 \section{Introduction}

In regression analysis, many analysts use the (penalized) least squares estimation, which aims at minimizing the mean squared prediction error \citep{Hastie:2009fg}.  On the other hand, the relative (percentage) error is often more useful and/or adequate than the mean squared error.
For example, in econometrics, the comparison of  prediction performance between different stock prices with different units should be made by relative error; we refer to \citep{Park:1998bk} and \citep{Ye:2007fz} among others.  Additionally, the prediction error of photovoltaic power production or electricity consumption is evaluated by not only mean squared error but also relative error (see, e.g., \citep{vanderMeer:2018bv}).
We refer to \citep{Mou16} regarding the usefulness and importance of the relative error.

In relative error estimation, we minimize a loss function based on the relative error.  An advantage of using such a loss function is that it is scale free or unit free.  Recently, several researchers have proposed various loss functions based on relative error \cite{Park:1998bk,Ye:2007fz,Chen:2010ic,Li:2013cm,Chen:2016ega,Ding:2017is}.  Some of these procedures have been extended to the nonparameteric model \citep{Demongeot:2016in} and  random effect model \citep{Wang:2018hn}.  The relative error estimation via the $L_1$ regularization, including the least absolute shrinkage and operator (lasso; \citep{tibshirani1996regression}), and the group lasso \citep{Yuan:2006jy}, have also been proposed by several authors \citep{Hao:2016ei,Liu:2016jy,Xia:2016bw}, to allow for the analysis of high-dimensional data.  

In practice, a response variable $y (> 0)$ can turn out to be extremely large or close to zero.  For example, the electricity consumption of a company may be low during holidays and high on exceptionally hot days. These responses may often be considered to be outliers, to which the relative error estimator is sensitive because the loss function diverges when $y \rightarrow \infty$ or $y \rightarrow 0$.  Therefore, a relative error estimation that is robust against outliers must be considered. Recently, \citet{Chen:2016ega} discussed the robustness of various relative error estimation procedures by investigating the  corresponding distributions,  and concluded that the distribution of least product relative error estimation (LPRE) proposed by \citep{Chen:2016ega} has heavier tails than others, implying that the LPRE might be more robust than others in practical applications.  However, our numerical experiments show that the LPRE is not as robust as expected, so that the robustification of the LPRE is yet to be investigated from the both theoretical and practical viewpoints.

To achieve a relative error estimation that is robust against outliers, this paper employs the $\gamma$-likelihood function for regression analysis by \citet{Kawashima:2017ho}, which is constructed by the $\gamma$-cross entropy \citep{Fujisawa:2008dq}.   The estimating equation is shown to have a redescending property, a desirable property in robust statistics literature \citep{maronna2006robust}. To find a minimizer of the negative $\gamma$-likelihood function, we construct a majorize-minimization (MM) algorithm. The loss function of our algorithm at each iteration is shown to be convex, although the original negative $\gamma$-likelihood function is nonconvex.   Our algorithm is guaranteed to decrease the objective function at each iteration.    
Moreover, we derive the asymptotic normality of the corresponding estimator together with a simple consistent estimator of the asymptotic covariance matrix, which enables us to straightforwardly create approximate confidence sets.  
Monte Carlo simulation is conducted to investigate the performance of our proposed procedure.  An  analysis of electricity consumption data is presented to illustrate the usefulness of our procedure.
 
The reminder of this paper is organized as follows: Section \ref{hk:sec_re.est} reviews several relative error estimation procedures.  In Section \ref{hk:sec_gam-div}, we propose a relative error estimation that is robust against outliers via the $\gamma$-likelihood function.  
Section \ref{hm:sec_theoretical.properties} presents theoretical properties: the redescending property of our method and the asymptotic distribution of the estimator, the proof of the latter being  deferred to Appendix \ref{hm:sec_proof}.
In Section \ref{sec:algorithm}, the MM algorithm is constructed to find the minimizer of the negative $\gamma$-likelihood function.   
Section \ref{sec:simulation} investigates the effectiveness of our proposed procedure via Monte Carlo simulations.  Section \ref{sec:real data analysis} presents the analysis on  electricity consumption data.  Finally, concluding remarks are given in Section \ref{sec:discussion}. 

\section{Relative error estimation}
\label{hk:sec_re.est}

Suppose that $\bm{x}_i$ = $(x_{i1},\dots,x_{ip})^T$ $(i=1,...,n)$ are predictors and $\bm{y} = (y_1, . . . , y_n )^T$ is a vector of positive responses.
Consider the multiplicative regression model
\begin{equation}
y_i = \exp(\bm{x}_i^T\bm{\beta} )\varepsilon_i = \exp\left( \sum_{j=1}^{p}x_{ij}\beta_j \right)\varepsilon_i,
\quad  \quad (i=1,\dots,n), \label{eq:multiplicative regression model}
\end{equation}
where $\bm{\beta} = (\beta_1,\cdots,\beta_p)^T$ is a $p$-dimensional coefficient vector, and $\varepsilon_i$ are positive random variables.
Predictors $\bm{x}_i\in\mbbr^p$ may be random and serially dependent, while we often set $x_{i1}= 1$, that is, incorporate the intercept in the exponent.
The parameter space $\mcb\subset\mbbr^p$ of $\bm{\beta}$ is a bounded convex domain such that $\bm{\beta}_0\in\mcb$.
We implicitly assume that the model is correctly specified, so that there exists a true parameter $\bm{\beta}_0=(\beta_{1,0},\dots,\beta_{p,0})\in\mcb$. We want to estimate $\bbeta_0$ from a sample $\{(\bm{x}_i,y_i)$, $i=1,\dots,n\}$.

We first remark that the condition $x_{i1}= 1$ ensures that the model \eqref{eq:multiplicative regression model} is scale-free regarding variables $\ve_i$, which is an essentially different nature from the linear regression model $y_i = \bm{x}_i^T\bm{\beta} +\ve_i$.
Specifically, multiplying a positive constant $\sig$ to $\ve_i$ results in the translation of the intercept in the exponent:
\begin{equation}
y_i = \exp(\bm{x}_i^T\bm{\beta} ) \sig\ve_i = \exp(\log \sig + \bm{x}_i^T\bm{\beta} ) \ve_i,
\nonumber
\end{equation}
so that the change from $\ve_i$ to $\sig\ve_i$ is equivalent to that from $\beta_1$ to $\beta_1+\log\sig$.  See Remark \ref{hm:remark_noise} on the distribution of $\ve_1$.

To provide a simple expression of the loss functions based on the relative error, we write
\begin{equation}
t_i = t_i(\bm{\beta}) = \exp(\bm{x}_i^T\bm{\beta}), \quad (i=1,\dots,n).
\nn
\end{equation}
Chen {\it et al.} \cite{Chen:2010ic,Chen:2016ega} pointed out that the loss criterion for relative error may depend on $| (y_i- t_i ) / y_i |$ and / or $| (y_i- t_i ) /t_i |$.  These authors also proposed general relative error (GRE) criteria, defined as
\begin{equation}
	G(\bm{\beta}) = \sum_{i=1}^n g\left(  \left| \frac{y_i- t_i}{y_i}  \right|,  \left|\frac{y_i- t_i}{t_i}\right| \right), \label{eq:GRE}
\end{equation}
where $g: [0,\infty)\times [0,\infty) \to [0,\infty)$.  Most of the loss functions based on the relative error are included in the GRE.   \citet{Park:1998bk} considered a loss function $g (a, b) = a^2$.
It may highly depend on a small $y_i$, because it includes $1/y_i^2$ terms, and then the estimator can be numerically unstable. 
Consistency and asymptotic normality may not be established under general regularity conditions \citep{Chen:2016ega}. 
The loss functions based on $g (a, b) = \max\{a, b\}$ \citep{Ye:2007fz} and $g(a,b) = a + b$ (least absolute relative error estimation, \citep{Chen:2010ic}) can have desirable asymptotic properties \citep{Ye:2007fz,Chen:2010ic}.  However, the minimization of the loss function can be challenging, in particular for high-dimensional data, when the function is nonsmooth or nonconvex.

In practice, the following two criteria would be useful:   
\begin{description}
	\item[Least product relative error estimation (LPRE)] \citet{Chen:2016ega} proposed the LPRE given by $g(a,b) = ab$.  
	The LPRE tries to minimize the product $| 1- t_i  / y_i |\times| 1- y_i  /t_i |$, not necessarily both terms at once.    
	\item[Least squared-sum relative error estimation (LSRE)] \citet{Chen:2016ega} considered the LSRE  given by $g(a,b) = a^2 + b^2$.  The LSRE aims to minimize both $| 1- t_i  / y_i |$ and $| 1- y_i  /t_i |$ through sum of squares $( 1- t_i  / y_i )^2 + ( 1- y_i  /t_i )^2$.
\end{description}
  The loss functions of LPRE and LSRE are smooth and convex, and also possess desirable asymptotic properties \citep{Chen:2016ega}.  The above-described GRE criteria and their properties are summarized in Table \ref{tab:criteria}.  Particularly, the ``convexity'' in the case of $g(a,b)=a+b$ a.s. holds when $\ve_i>0$, $\ve_i\ne 1$, and $\sumi\bm{x}_{i}\bm{x}_i^{T}$ is positive definite, since the Hessian matrix of the corresponding $G(\bbeta)$ is $\sumi |\ve_i-\ve_i^{-1}| \bm{x}_i\bm{x}_i^{T}$ a.s.

\begin{table}[t]
\caption{
Several examples of GRE criteria and their properties.  ``Likelihood'' in the second column means the existence of a likelihood function that corresponds to the loss function.
The properties of ``Convexity'' and ``Smoothness'' in the last two columns respectively indicate those with respect to $\bm{\beta}$ of the corresponding loss function.
}
\label{tab:criteria}
\centering
\begin{tabular}{cccc}
\hline
\textbf{$g(a,b)$} & {Likelihood} & {Convexity} & {Smoothness} \\
\hline
$a^2$		        &			            &               & $\surd$ \\
$a + b$		    &	$\surd$		&       $\surd$        &               \\
$\max\{a,b\}$	&	$\surd$		&               &               \\
$ab$					&  $\surd$       & $\surd$  & $\surd$ \\
$a^2 +  b^2$	&	$\surd$		& $\surd$  & $\surd$ \\
\hline
\end{tabular}
\end{table}

Although not essential, we assume that the variables $\ve_i$ in \eqref{eq:multiplicative regression model} are i.i.d. with common density function $h$.  This implies the conditional distribution $\mcl(y_i|\bm{x}_i)$ is homogeneous in the sense that it does not depend on $i$.
As in \citet{Chen:2016ega}, we consider the following class of $h$ associated with $g$:
\begin{equation}
h(\varepsilon) := \frac{C(g)}{\varepsilon} \exp\left\{ -\rho(\ve) \right\}I_{+}(\varepsilon),
\label{ep.pdf}
\end{equation}
where
\begin{equation}
\rho(\ve)=\rho(\ve;g) := g\left( \left|1-\frac{1}{\varepsilon} \right|,\, |1-\varepsilon|\right)
\nonumber
\end{equation}
and $C(g)$ is a normalizing constant ($\int h(\varepsilon)d\varepsilon=1$) and $I_+$ denotes the indicator function of set $(0,\infty)$.
Further, we assume the symmetry property $g(a,b)=g(b,a)$, $a,b\ge 0$, from which it follows that $\ve_1 \sim \ve_1^{-1}$.
The latter property is necessary for a score function to be associated with the gradient of a GRE loss function, hence being a martingale with respect to a suitable filtration, which often entails estimation efficiency. Indeed, the asymmetry of $g(a,b)$ (i.e., $g(a,b) \neq g(b,a)$) may produce a substantial bias in the estimation \citep{Ye:2007fz}.
The entire set of our regularity conditions will be shown in Section \ref{sec:asymptotics}. The conditions therein concerning $g$ are easily verified for both LPRE and LSRE.

The density function of response $y$ given $\bm{x}_i$, say $f(y|\bm{x}_i;\bm{\beta})$, is 
\begin{eqnarray}
	f(y|\bm{x}_i;\bm{\beta}) &=& \exp(-\bm{x}_i^T\bm{\beta}) h\left(y \exp(-\bm{x}_i^T\bm{\beta})\right) \label{hm:add.eq-1} \\
	&=& \frac{1}{t_i} h\left( \frac{y}{t_i} \right). \label{eq:density}
\end{eqnarray} 
From \eqref{ep.pdf} we see that the maximum likelihood estimator (MLE) based on the error distribution in (\ref{eq:density}) is obtained by the minimization of (\ref{eq:GRE}). 
 For example, the density functions of LPRE and LSRE are  
\begin{eqnarray}
	LPRE: && f(y|\bm{x}_i)  = \frac{1}{2K_0(2)} y^{-1}\exp \left( -\frac{y}{t_i} -\frac{t_i}{y}  \right) \quad  y>0, 
	\label{hm:LPRE_pdf} \\
	LSRE: && f(y|\bm{x}_i)  = C_{LSRE} y^{-1} \exp \left\{ - \left(1-\frac{t_i}{y}\right)^2  -\left(1-\frac{y}{t_i}\right)^2   \right\} \quad  y>0, \nn
	\end{eqnarray}
where $K_{\nu}(z)$ denotes a modified Bessel function of third kind with index $\nu\in\mbbr$:
\begin{eqnarray*}
K_{\nu}(z) = \frac{z^{\nu}}{2^{\nu+1}}\int_{0}^{\infty} t^{-\nu-1} \exp \left( -t-\frac{z^2}{4t}  \right) dt,
\end{eqnarray*}
and $C_{LSRE}$ is a constant term.  Constant terms are numerically computed as $K_0(2) \approx 0.1139$ and $C_{LSRE} \approx 0.911411$.
Density \eqref{hm:LPRE_pdf} is a special case of the generalized inverse Gaussian distribution (see, e.g., \citep{Koudou:2014ja}).

\begin{remark}
\label{hm:remark_noise}
We assume that the noise density $h$ is fully specified in the sense that, given $g$, the density $h$ does not involve any unknown quantity.
However, this is never essential.
For example, for the LPRE defined by \eqref{hm:LPRE_pdf} we could naturally incorporate one more parameter $\sig>0$ into $h$, the resulting form of $h(\ve)$ being
\begin{equation}
\ve \mapsto \frac{1}{2K_0(\sig)} \ve^{-1}\exp\left\{ -\frac{\sig}{2}\left( \ve+\frac{1}{\ve}\right)\right\}I_{+}(\ve).
\nonumber
\end{equation}
Then, we can verify that the distributional equivalence $\ve_1\sim\ve_1^{-1}$ holds whatever the value of $\sig$ is.
Particularly, the estimation of parameter $\sig$ does make statistical sense and, indeed, it is possible to deduce the asymptotic normality of the joint maximum-(partial-)likelihood estimator of $(\bbeta,\sig)$.
In this paper, we do not pay attention to such a possible additional parameter, but instead regard it (whenever it exists) as a nuisance parameter, as in the noise variance in the least-squares estimation of a linear regression model.
\end{remark}

\section{Robust estimation via $\gamma$-likelihood}
\label{hk:sec_gam-div}

In practice, outliers can often be observed. For example, the electricity consumption data can have the outliers on extremely hot days.  The estimation methods via GRE criteria, including LPRE and LSRE, are not robust against outliers, because the corresponding density functions are not generally heavy-tailed.  Therefore, a relative error estimation method that is robust against the outliers is needed.  
To achieve this, we consider minimizing the negative $\gamma$-(partial-)likelihood function based on the $\gam$-cross entropy \citep{Kawashima:2017ho}.

We assume the conditional density of $y_i$ given $\{(\bm{x}_i,\bm{x}_{i-1},\bm{x}_{i-2},\dots),\,(y_{i-1},y_{i-2},\dots)\}$ equals $f(\cdot|\bm{x}_i;\bbeta)$.  Then, we define the negative $\gamma$-(partial-)likelihood function by
\begin{eqnarray}
\ell_{\gamma,n}(\bm{\beta}) = -\frac{1}{\gamma} \log \left\{ \frac{1}{n}\sum_{i=1}^n f(y_i|\bm{x}_i;\bm{\beta})^{\gamma}  \right\} 
+ \frac{1}{1+\gamma} \log \left\{ \frac{1}{n}\sum_{i=1}^n\int_{0}^{\infty} f(y|\bm{x}_i;\bm{\beta})^{1+\gamma}dy \right\},
\label{gamma-div}
\end{eqnarray}
where $\gamma>0$ is a parameter that controls the degrees of robustness;  $\gamma \rightarrow 0$ corresponds to the negative log-likelihood function, and robustness is enhanced as $\gamma$ increases.   On the other hand, a too large $\gamma$ can decrease the efficiency of the estimator \citep{Fujisawa:2008dq}.  
In practice, the value of $\gamma$ may be selected by a cross-validation based on $\gamma$-cross entropy (see, e.g., \citep{JonHjoHarBas01,Fujisawa:2008dq}).  We refer to \citet{KawFuj.arxiv18} for more recent observations on comparison of the $\gam$-divergences between \citet{Fujisawa:2008dq} and \citet{Kawashima:2017ho}.

The integration $\int f(y|\bm{x}_i;\bm{\beta})^{1+\gamma}dy $ in the second term on the right-hand side of (\ref{gamma-div}) is 
\begin{align}
\int_0^\infty f(y|\bm{x}_i;\bm{\beta})^{1+\gamma}dy
&=  \frac{1}{t_i^{1+\gamma}} \int_0^\infty \left\{h \left( \frac{y}{t_i} \right) \right\}^{1+\gamma}dy
=: t_i^{-\gamma} C(\gamma,h),
\nonumber
\end{align}
where
\begin{equation}
C(\gamma,h) := \int_0^\infty h(v)^{1+\gamma}dv
\label{hm:C_def}
\end{equation}
is a constant term, which is assumed to be finite. Then, (\ref{gamma-div}) is expressed as
\begin{eqnarray}
 \ell_{\gamma,n}(\bm{\beta}) 
 &=& \mathop{\underbrace{-\frac{1}{\gamma} \log \left\{\sum_{i=1}^n f(y_i|\bm{x}_i;\bm{\beta})^{\gamma}  \right\}}}_{=:\,\ell_1(\bm{\beta})} 
 + \mathop{\underbrace{\frac{1}{1+\gamma} \log \left\{\sum_{i=1}^n t_i^{-\gamma} \right\}}}_{=:\,\ell_2(\bm{\beta})}
  + C_0(\gamma,h) , \label{eq:gamma-likelihood}
\end{eqnarray}
where $C_0(\gamma,h)$ is a constant term free from $\bm{\beta}$.  We define the maximum $\gam$-likelihood estimator to be any element such that
\begin{equation}
\bhg \in \argmin   \ell_{\gamma,n} .
\label{hm:gam.mle}
\end{equation}

\section{Theoretical properties}
\label{hm:sec_theoretical.properties}

\subsection{Technical assumptions}

Let $\cip$ denote the convergence in probability.

\begin{Assumption}[Stability of the predictor]
\label{hm:A-x}
There exists a probability measure $\pi(d\bm{x})$ on the state space $\mcx$ of the predictors and positive constants $\del,\del'>0$ such that
\begin{equation}
\frac1n \sumi |\bm{x}_i|^{3}\exp\left(\del' |\bm{x}_i|^{1+\del}\right) = O_p(1),
\nonumber
\end{equation}
and that
\begin{equation}
\frac{1}{n}\sumi \eta(\bm{x}_i) \cip \int_{\mcx}\eta(\bm{x})\pi(d\bm{x}), \qquad n\to\infty,
\nonumber
\end{equation}
the limit being finite for any measurable $\eta$ satisfying that
\begin{equation}
\sup_{\bm{x}\in\mbbr^{p}}\frac{|\eta(\bm{x})|}{(1+|\bm{x}|^{3})\exp\left(\del' |\bm{x}|^{1+\del}\right)} <\infty.
\nonumber
\end{equation}
\end{Assumption}

\begin{Assumption}[Noise structure]
\label{hm:A-ve}
The a.s. positive i.i.d. random variables $\ve_1,\ve_2,\dots$ have a common positive density $h$ of the form \eqref{ep.pdf}:
\begin{equation}
h(\ve) = \frac{C(g)}{\varepsilon} \exp\left\{ -\rho(\ve) \right\}I_{+}(\varepsilon),
\nonumber
\end{equation}
for which the following conditions hold.
\begin{enumerate}
\item Function $g:\,[0,\infty)\times [0,\infty) \to [0,\infty)$ is three times continuously differentiable on $(0,\infty)$ and satisfies that
\begin{equation}
g(a,b)=g(b,a), \qquad a,b\ge 0.
\nonumber
\end{equation}

\item There exist constants $\kappa_{0}, \kappa_{\infty}>0$, and $c>1$ such that
\begin{equation}
\frac{1}{c}\left( \ve^{-\kappa_0} \vee \ve^{\kappa_{\infty}}\right) \le 
\rho(\ve) 
\le c\left( \ve^{-\kappa_0} \vee \ve^{\kappa_{\infty}}\right)
\nonumber
\end{equation}
for every $\ve>0$.

\item There exist constants $c_0, c_\infty \ge 0$ such that
\begin{equation}
\sup_{\ve>0} \left( \ve^{-c_0} \vee \ve^{c_\infty} \right)^{-1} 
\max_{k=1,2,3} \left| \p_{\ve}^{k}\rho(\ve) \right| < \infty.
\nonumber
\end{equation}

\end{enumerate}
\end{Assumption}
Here and in the sequel, for a variable $a$, we denote by $\p_{a}^{k}$ the $k$th-order partial differentiation with respect to $a$.

Assumption \ref{hm:A-x} is necessary to identify the large-sample stochastic limits of the several key quantities in the proofs:
without them, we will not be able to deduce an explicit asymptotic normality result.
Assumption \ref{hm:A-ve} holds for many cases, including the LPRE and the LSRE (i.e. $g(a,b)=ab$ and $a^2+b^2$), while excluding $g(a,b)=a^2$ and $g(a,b)=b^2$.
The smoothness condition on $h$ on $(0,\infty)$ is not essential and could be weakened in light of the $M$-estimation theory \citep[Chapter 5]{vdV98}.
Under these assumptions, we can deduce the following statements.
\begin{itemize}
\item $h$ is three times continuously differentiable on $(0,\infty)$, and for each $\al>0$,
\begin{equation}
\int_0^\infty h^{\al}(\ve)d\ve < \infty \quad \text{and}\quad \max_{k=0,1,2,3}\sup_{\ve>0}\left| \p_{\ve}^{k}\left\{h(\ve)^{\al}\right\} \right| < \infty.
\nonumber
\end{equation}
\item For each $\gam>0$ and $\al>0$ (recall that the value of $\gam>0$ is given),
\begin{equation}
\lim_{\ve\downarrow 0}h(\ve)^{\gam}\left| u_h(\ve) \right|^{\al} = \lim_{\ve\uparrow \infty}h(\ve)^{\gam}\left| u_h(\ve) \right|^{\al} = 0,
\label{hm:redescending+1}
\end{equation}
where
\begin{equation}
u_h(z) := 1+z\, \p_z \log h(z) = 1+z \frac{h'(z)}{h(z)}.
\nn
\end{equation}
\end{itemize}
The verifications are straightforward hence omitted.

Finally, we impose
\begin{Assumption}[Identifiability]
\label{hm:A-iden}
We have $\bm{\beta}=\bm{\beta}_0$ if
\begin{equation}
\rho\big( e^{-\bm{x}^{T}\bbeta} y\big) = \rho\big( e^{-\bm{x}^{T}\bbeta_0} y\big) \qquad \text{$\pi(d\bm{x})\otimes\lam_+(dy)$-a.e. $(\bm{x},y)$},
\nonumber
\end{equation}
where $\lam_+$ denotes the Lebesgue measure on $(0,\infty)$.
\end{Assumption}

\subsection{Redescending property}\label{sec:Redescending property}

The estimating function based on the negative $\gamma$-likelihood function is given by
\begin{equation}
\sum_{i=1}^n \bm{\psi}({y}_i|\bm{x}_i;\bm{\beta})=\bm{0}.
\nonumber
\end{equation}
In our model, we consider not only too large $y_i$s but also too small $y_i$s as outliers:
the estimating equation has the redescending property if
\begin{align}
\lim_{y \rightarrow \infty} \bm{\psi}(y|\bm{x};\bm{\beta}_0) =
\lim_{y \rightarrow +0} \bm{\psi}(y|\bm{x};\bm{\beta}_0)  = \bm{0}
\nonumber
\end{align}
for each $\bm{x}$.  The redescending property is known as a desirable property in robust statistics literature \citep{maronna2006robust}.  Here, we show the proposed procedure has the redescending property.

The estimating equation based on the negative $\gamma$-likelihood function is 
\begin{eqnarray}
-\dfrac{\sum_{i=1}^n f(y_i|\bm{x}_i;\bm{\beta})^\gamma \bm{s} (y_i|\bm{x}_i;\bm{\beta})}{\sum_{j=1}^n f(y_j|\bm{x}_j;\bm{\beta})^\gamma } + \frac{\partial}{\partial\bm{\beta}} \ell_2(\bm{\beta}) = \bm{0}, 
\nn
\end{eqnarray}
where
$$  \bm{s} (y|\bm{x};\bm{\beta}) = \frac{\partial \log f(y|\bm{x};\bm{\beta})}{ \partial \bm{\beta}}.
$$
We have  expression
\begin{eqnarray}
\bm{\psi}(y|\bm{x};\bm{\beta})  = f(y|\bm{x};\bm{\beta})^\gamma \left\{\bm{s} (y|\bm{x};\bm{\beta}) - \frac{\partial}{\partial\bm{\beta}} \ell_2(\bm{\beta}) \right\}. 
\nn
\end{eqnarray}
Note that $\frac{\partial}{\partial\bm{\beta}} \ell_2(\bm{\beta})$ is free from $y$.
For each $(\bm{x},\bbeta)$, direct computations give the estimate
\begin{equation}
\left| \bm{\psi}(y|\bm{x};\bm{\beta}) \right|
\le C(\bm{x};\bbeta) h\left( \exp(-\bm{x}^{T}\bbeta)y \right)^{\gam} \left|u_h\left( \exp(-\bm{x}^{T}\bbeta)y \right)\right|
\label{hm:rdp_ineq1}
\end{equation}
for some constant $C(\bm{x};\bbeta)$ free from $y$.
Hence, \eqref{hm:redescending+1} combined with the inequality \eqref{hm:rdp_ineq1} leads to the redescending property.

\subsection{Asymptotic distribution}\label{sec:asymptotics}

Recall \eqref{hm:C_def} for the definition of $C(\gam,h)$ and let
\begin{align}
C_1(\gam,h) &:= \int_0^\infty \ve h(\ve)^{\gam}h'(\ve)d\ve, \nn\\
C_2(\gam,h) &:= \int_0^\infty u_h(\ve)^{2} h(\ve)^{2\gam+1}d\ve, \nn\\
\Pi_k(\gam) &:= \int \bm{x}^{\otimes k}\exp(- \gam \bm{x}^{T}\bbeta_0)\pi(d\bm{x}), \qquad k=0,1,2,
\nonumber
\end{align}
where $\bm{x}^{\otimes 0}:=1\in\mbbr$, $\bm{x}^{\otimes 1}:=\bm{x}\in\mbbr^p$, and $\bm{x}^{\otimes 2}:=\bm{x}\bm{x}^{T}\in\mbbr^p\otimes\mbbr^p$;
Assumptions \ref{hm:A-x} and \ref{hm:A-ve} ensure that all these quantities are finite for each $\gam>0$. Moreover,
\begin{align}
H'_\gam(\bbeta_0) &:= \iint f(y|\bm{x};\bbeta_0)^{\gam+1}dy\pi(d\bm{x})
= C(\gam,h) \Pi_0(\gam), \nn\\
H''_\gam(\bbeta_0) &:= \iint f(y|\bm{x};\bbeta_0)^{\gam+1}s(y|\bm{x};\bbeta_0)dy\pi(d\bm{x})
= -\left\{ C(\gam,h)+C_1(\gam,h) \right\} \Pi_1(\gam), \nn\\
\D_\gam(\bbeta_0) &:= C(\gam,h)^2 C_2(\gam,h) \Pi_0(\gam)^2 \Pi_2(2\gam) \nn\\
&{}\qquad + \left\{ C(\gam,h)+C_1(\gam,h) \right\}^2 C(2\gam,h) \Pi_0(2\gam) \Pi_1(\gam)^{\otimes 2} \nn\\
&{}\qquad -2 C(\gam,h) \left\{ C(\gam,h)+C_1(\gam,h) \right\} \left\{ C(2\gam,h)+C_1(2\gam,h) \right\} \Pi_0(\gam) \Pi_1(2\gam) \Pi_1(\gam)^{T},
\label{hm:D_def} \\
J_\gam(\bbeta_0) &:= 
C(\gam,h)
C_2(\gam/2,h)
\Pi_0(\gam)\Pi_2(\gam) - \left\{ C(\gam,h)+C_1(\gam,h) \right\}^2 \Pi_1(\gam)^{\otimes 2}.
\label{hm:J_def}
\end{align}
We are assuming that density $h$ and tuning parameter $\gam$ are given \textit{a priori}, hence we can (numerically) compute constants $C(\gam,h)$, $C_1(\gam,h)$, and $C_2(\gam,h)$.
In the following, we often omit ``$(\bbeta_0)$'' from the notation. 

Let $\cil$ denote the convergence in distribution.

\begin{Theorem}
Under Assumptions \ref{hm:A-x} -- \ref{hm:A-iden}, we have
\begin{equation}
\sqrt{n}\left(\hat{\bm{\beta}}_{\gam}-\bm{\beta}_0 \right) \cil N_p\left(\bm{0},\, J_{\gam}^{-1}\Delta_{\gam} J_{\gam}^{-1}\right).
\label{hm:thm1-an}
\end{equation}
The asymptotic covariance matrix can be consistently estimated through expressions \eqref{hm:D_def} and \eqref{hm:J_def} with quantities $\Pi_{k}(\gam)$ therein replaced by the empirical estimates:
\begin{equation}
\hat{\Pi}_{k,n}(\gam) := \frac1n \sumi \bm{x}_i^{\otimes k}\exp(- \gam \bm{x}_i^{T}\hat{\bbeta}_{\gam})
\cip \Pi_{k}(\gam), \qquad k=0,1,2.
\label{hm:thm1-av.ce}
\end{equation}
\label{hm:thm1}
\end{Theorem}

The proof of Theorem \ref{hm:thm1} will be given in Appendix \ref{hm:sec_proof}.
Note that, for $\gam\to 0$, we have $C(\gam,h)\to 1$, $C_1(\gam,h) \to -1$, and $C_2(\gam,h) \to \int_0^\infty u_h(\ve)^{2}h(\ve)d\ve$,
which in particular entails $H'_{\gam}\to 1$ and $H''_{\gam}\to \bm{0}$.
Then, both $\D_\gam$ and $J_\gam$ tend to the Fisher information matrix
\begin{equation}
\mci_0 := \iint s(y|\bm{x};\bm{\beta}_0)^{\otimes 2} f(y|\bm{x},\bbeta_0)\pi(d\bm{x})dy
=\int_0^\infty u_h(\ve)^{2}h(\ve)d\ve \, \int \bm{x}^{\otimes 2}\pi(d\bm{x})
\nonumber
\end{equation}
as $\gam\to 0$, so that the asymptotic distribution $N_p(\bm{0},\, J_{\gam}^{-1}\Delta_{\gam} J_{\gam}^{-1})$ becomes $N_p(\bm{0},\, \mci_0^{-1})$, the usual one of the MLE.  We also note that, without details, we could deduce a density-power divergence counterpart to Theorem \ref{hm:thm1} similarly but with slightly lesser computation cost, in which case, we consider the objective function
\begin{eqnarray}
\ell_{\gamma,n}(\bm{\beta}) = -\frac{1}{\gamma} \frac{1}{n}\sum_{i=1}^n f(y_i|\bm{x}_i;\bm{\beta})^{\gamma} 
+ \frac{1}{1+\gamma} \frac{1}{n}\sum_{i=1}^n\int_{0}^{\infty} f(y|\bm{x}_i;\bm{\beta})^{1+\gamma}dy
\nn
\end{eqnarray}
instead of the $\gam$-(partial-)likelihood \eqref{gamma-div}.
See \citet{BHHJ98} and \citet{JonHjoHarBas01} for details of the density-power divergence, also known as the $\beta$-divergence \citep{EguKan01}.


\section{Algorithm}\label{sec:algorithm}
Even if the GRE criterion in (\ref{eq:GRE}) is a convex function, the negative $\gamma$-likelihood function is nonconvex.  Therefore, it is difficult to find a global minimum.   Here, we derive the MM (majorize-minimization) algorithm to obtain a local minimum.  The MM algorithm monotonically decreases the objective function at each iteration.
We refer to  \citet{hunter2004tutorial} for a concise account of the MM algorithm.

Let $\bm{\beta}^{(t)}$ be the value of the parameter at the $t$th iteration.  The negative $\gamma$-likelihood function in (\ref{eq:gamma-likelihood}) consists of two nonconvex functions, $\ell_1(\bm{\beta})$ and $\ell_2(\bm{\beta})$.    The majorization functions of $\ell_j(\bm{\beta})$, say $\tilde{\ell}_j(\bm{\beta}|\bm{\beta}^{(t)})$ ($j=1,2$), are constructed so that the optimization of $\min_{\bm{\beta}}\tilde{\ell}_j(\bm{\beta}|\bm{\beta}^{(t)})$ is much easier than that of $\min_{\bm{\beta}}\ell_j(\bm{\beta})$.  The majorization functions must satisfy the following inequalities:
\begin{eqnarray}
\tilde{\ell}_j(\bm{\beta}|\bm{\beta}^{(t)}) &\ge& \ell_j(\bm{\beta}),\label{MM_p1}\\
\tilde{\ell}_j(\bm{\beta}^{(t)}|\bm{\beta}^{(t)}) &=& \ell_j(\bm{\beta}^{(t)}).\label{MM_p2}
\end{eqnarray}
Here we construct majorization functions $\tilde{\ell}_j(\bm{\beta}|\bm{\beta}^{(t)})$ for $j=1,2$.
\subsection{Majorization function for $\ell_1(\bm{\beta})$}\label{subsec:MM1}

Let
\begin{eqnarray}
w_i^{(t)} &=& \frac{f(y_i|\bm{x}_i;\bm{\beta}^{(t)})^{\gamma} }{\sum_{j=1}^nf(y_j|\bm{x}_j;\bm{\beta}^{(t)})^{\gamma}}, \label{weight_gamma}\\
r_i^{(t)} &=& \sum_{j=1}^n f(y_j|\bm{x}_j;\bm{\beta}^{(t)})^{\gamma} \dfrac{ f(y_i|\bm{x}_i;\bm{\beta})^{\gamma}}{ f(y_i|\bm{x}_i;\bm{\beta}^{(t)})^{\gamma}}.\label{ri}
\end{eqnarray}
Obviously, $\sum_{i=1}^n w_i^{(t)} = 1$ and $w_i^{(t)} r_i^{(t)} =  f(y_i| \bm{x}_i;\bm{\beta})^{\gamma}$.
Applying Jensen's inequality to $y = -\log x$, we obtain inequality
\begin{eqnarray}
-\log\left(\sum_{i=1}^nw^{(t)}_i r_i^{(t)}\right) \le -\sum_{i=1}^n w^{(t)}_i\log r_i^{(t)}.
\label{jensen}
\end{eqnarray}
Substituting (\ref{weight_gamma}) and (\ref{ri}) into (\ref{jensen}) gives  
\begin{eqnarray}
\ell_1(\bm{\beta}) \le - \sum_{i=1}^n  w_i^{(t)} \log f(y_i| \bm{x}_i;\bm{\beta}) + C, 
\nn
\end{eqnarray}
where $C = \frac{1}{\gamma} \sum_i w_i^{(t)}\log w_i^{(t)} $. Denoting
\begin{equation}
\tilde{\ell}_1(\bm{\beta}|\bm{\beta}^{(t)}) = -  \sum_{i=1}^n  w_i^{(t)} \log f( y_i|\bm{x}_i;\bm{\beta}) + C, \label{eq:l1tilde}
\end{equation}
we observe that  (\ref{eq:l1tilde}) satisfies (\ref{MM_p1}) and (\ref{MM_p2}).  It is shown that  $\tilde{\ell}_1(\bm{\beta}|\bm{\beta}^{(t)})$ is a convex function if the original relative error loss function is convex.  Particularly, the majorization functions $\tilde{\ell}_1(\bm{\beta}|\bm{\beta}^{(t)})$ based on LPRE and LSRE are both convex.  

\subsection{Majorization function for $\ell_2(\bm{\beta})$}\label{subsec:MM2}
Let $\theta_i = -\gamma\bm{x}_i^T\bm{\beta}$.  We view $\ell_2(\bm{\beta})$ as a function of  $\bm{\theta} = (\theta_1,\dots,\theta_n)^T$.  Let
\begin{eqnarray}
	s(\bm{\theta}) := \log\left( \sum_{i=1}^nt_i^{-\gamma}\right) = \log\left( \sum_{i=1}^n \exp(\theta_i)\right). \label{eq:stheta}
\end{eqnarray}
By taking the derivative of $s(\bm{\theta})$ with respect to $\bm{\theta}$, we have
\begin{eqnarray*}
	\frac{\partial s(\bm{\theta})}{\partial \theta_i} = \pi_i,	\quad \frac{\partial^2 s(\bm{\theta})}{\partial \theta_j \partial \theta_i} = \pi_i\delta_{ij} - \pi_i\pi_j,
\end{eqnarray*}
where $\pi_i = \exp(\theta_i)  / \{\sum_{k=1}^n \exp(\theta_k) \}$. 
Note that $\sum_{i=1}^n \pi_i=1$ for any $\bm{\theta}$.

The Taylor expansion of $s(\bm{\theta})$ at $\bm{\theta} = \bm{\theta}^{(t)}$ is expressed as
\begin{eqnarray}
	s(\bm{\theta}) &=& s(\bm{\theta}^{(t)}) + \bm{\pi}^{(t)T}  (\bm{\theta} - \bm{\theta}^{(t)}) + \frac{1}{2}(\bm{\theta} - \bm{\theta}^{(t)})^T \frac{\partial^2 s(\bm{\theta}^*)}{\partial \bm{\theta} \partial \bm{\theta}^T}  (\bm{\theta} - \bm{\theta}^{(t)}), \label{eq:taylor}
\end{eqnarray}
where $\bm{\pi}^{(t)} = (\pi_1^{(t)},\dots,\pi_n^{(t)})^T$ and $\bm{\theta}^*$ is an $n$-dimensional vector located between $\bm{\theta}$ and $\bm{\theta}^{(t)}$.  We define an $n \times n$ matrix $\bm{B}$ as follows:
\begin{equation}
\bm{B} := \frac{1}{2} \left(\bm{I} - \frac{1}{n}\bm{1}\bm{1}^T\right).
\nn
\end{equation}
It follows from \citep{Bhning:1992fg} that, in the matrix sense,
\begin{equation}
\frac{\partial^2 s(\bm{\theta})}{\partial \bm{\theta} \partial \bm{\theta}^T} \leq \bm{B}
\label{MMlogistic}
\end{equation}
for any $\bm{\theta}$. Combining (\ref{eq:taylor}) and (\ref{MMlogistic}), we have
\begin{eqnarray}
	s(\bm{\theta})	& \leq & s(\bm{\theta}^{(t)}) + \bm{\pi}^{(t)T}  (\bm{\theta} - \bm{\theta}^{(t)}) + \frac{1}{2} (\bm{\theta} - \bm{\theta}^{(t)})^T \bm{B} (\bm{\theta} - \bm{\theta}^{(t)}). \label{eq:ineqs}
	\end{eqnarray}
Substituting (\ref{eq:stheta}) into (\ref{eq:ineqs}) gives
\begin{eqnarray*}
	\log\left\{ \sum_{i=1}^n  \exp(-\gamma\bm{x}_i^T\bm{\beta})\right\} & \leq & \log\left\{ \sum_{i=1}^n  \exp(-\gamma\bm{x}_i^T\bm{\beta}^{(t)})\right\}  - \gamma\bm{\pi}^{(t)T}  \bm{X}(\bm{\beta} - \bm{\beta}^{(t)}) \\
	&& + \frac{\gamma^2}{2} (\bm{\beta} - \bm{\beta}^{(t)})^T  \bm{X}^T\bm{B} \bm{X}(\bm{\beta} - \bm{\beta}^{(t)}),
	\end{eqnarray*}
where  $\bm{X} = (\bm{x}_1,\dots,\bm{x}_n)^T$.  The majorization function of $\ell_2(\bm{\beta} )$ is then constructed by 
\begin{eqnarray}
\tilde{\ell}_2(\bm{\beta} | \bm{\beta}^{(t)}) &=& \frac{\gamma^2}{2(1+\gamma)} \bm{\beta}^T  \bm{X}^T\bm{B} \bm{X}\bm{\beta}
 - \frac{\gamma}{1+\gamma}\bm{\beta}^T (\bm{X}^T\bm{\pi}^{(t)} + \gamma \bm{X}^T\bm{B} \bm{X} \bm{\beta}^{(t)} ) + C, \label{eq:l2tilde}
 \end{eqnarray}
where $C$ is a constant term free from $\bbeta$.  We observe that $\tilde{\ell}_2(\bm{\beta} | \bm{\beta}^{(t)})$ in  (\ref{eq:l2tilde}) satisfies (\ref{MM_p1}) and (\ref{MM_p2}).   It is  shown that $\tilde{\ell}_2(\bm{\beta} | \bm{\beta}^{(t)})$ is a convex function, because $\bm{X}^T\bm{B}\bm{X}$ is positive semi-definite.  

\subsection{MM algorithm for robust relative error estimation}
In Subsections \ref{subsec:MM1} and  \ref{subsec:MM2}, we have constructed the majorization functions for both $\ell_1(\bm{\beta}) $ and $\ell_2(\bm{\beta} )$.  The MM algorithm based on these majorization functions is detailed in Algorithm \ref{algorithm figure}.  The majorization function $\tilde{\ell}_1(\bm{\beta} | \bm{\beta}^{(t)}) + \tilde{\ell}_2(\bm{\beta} | \bm{\beta}^{(t)})$ is convex if the original relative error loss function is convex.  Particularly, the majorization functions of LPRE and LSRE are both convex. 
\begin{algorithm}[t]
\caption{Algorithm of robust relative error estimation.}
\label{algorithm figure}
\begin{algorithmic}[1]
\STATE $t \leftarrow 0$
\STATE Set an initial value of parameter vector $\bm{\beta}^{(0)}$.
\WHILE{$\bm{\beta}^{(t)}$ is converged}
\STATE Update the weights by (\ref{weight_gamma})
\STATE   Update $\bm{\beta}$ by
\begin{equation*}
	\bm{\beta}^{(t+1)} \leftarrow {\rm arg} \min_{\bm{\beta}} \{ \tilde{\ell}_1(\bm{\beta} | \bm{\beta}^{(t)})  +  \tilde{\ell}_2(\bm{\beta} | \bm{\beta}^{(t)})  \},
\end{equation*}
where $\tilde{\ell}_1(\bm{\beta} | \bm{\beta}^{(t)})$ and $\tilde{\ell}_2(\bm{\beta} | \bm{\beta}^{(t)})$ are given by (\ref{eq:l1tilde}) and (\ref{eq:l2tilde}), respectively.
\STATE  $t \leftarrow t + 1$
\ENDWHILE
\vspace{-1mm}
\end{algorithmic}
\end{algorithm}

\begin{remark}
	Instead of the MM algorithm, one can directly use the quasi-Newton method, such as the BFGS algorithm,  to minimize the negative $\gamma$-likelihood function.  In our experience, the BFGS algorithm is faster than the MM algorithm, but is more sensitive to an initial value than the MM algorithm.  The strengths of BFGS and MM algorithms would be shared by using the following hybrid algorithm:
	\begin{enumerate}
		\item We first conduct the MM algorithm with a small number of iterations.
		\item Then, the BFGS algorithm is conducted.  We use the estimate obtained by the MM algorithm as an initial value of the BFGS algorithm.
	\end{enumerate}    
	The stableness of the MM algorithm is investigated through the  real data analysis in Section \ref{sec:real data analysis}.  
\end{remark}
\begin{remark}\label{remark:lasso}
To deal with high-dimensional data, we often use the $L_1$ regulzarization, such as the lasso \citep{tibshirani1996regression}, elastic net \citep{zou2005regularization}, and SCAD \citep{fan2001variable}.  In robust relative error estimation, the loss function based on the lasso is expressed as
\begin{equation}
\ell_{\gamma,n}(\bm{\beta}) + \lambda\sum_{j=1}^p|\beta_j|, \label{eq:lasso}
\end{equation}	
where $\lambda>0$ is a regularization parameter.  However, the loss function in  (\ref{eq:lasso}) is non-convex and indifferentiable.  Instead of directly minimizing the non-convex loss function in (\ref{eq:lasso}), we may use the MM algorithm; the following convex loss function is minimized at each iteration:
\begin{equation}
\tilde{\ell}_1(\bm{\beta} | \bm{\beta}^{(t)})  +  \tilde{\ell}_2(\bm{\beta} | \bm{\beta}^{(t)})  
+ \lambda\sum_{j=1}^p|\beta_j|. \label{eq:lassoMM}
\end{equation}	
The minimization of (\ref{eq:lassoMM}) can be realized by the alternating direction method of multipliers algorithm \citep{Hao:2016ei} or the coordinate descent algorithm with quadratic approximation of $\tilde{\ell}_1(\bm{\beta} | \bm{\beta}^{(t)}) + \tilde{\ell}_2(\bm{\beta} | \bm{\beta}^{(t)})$ \citep{Friedman:2010cg}.  
\end{remark}

\section{Monte Carlo simulation}\label{sec:simulation}
\subsection{Setting}
We consider the following two simulation models as follows:
\begin{eqnarray*}
\hspace{-30mm} \mbox{\bf Model 1: }&& \bm{\beta}_0 = (1,1,1)^T,\\
\hspace{-30mm} \mbox{\bf Model 2: }&& \bm{\beta}_0 = (\underbrace{0.5, \cdots, 0.5}_{\mbox{\small 6}},  \underbrace{0 , \cdots , 0}_{\mbox{\small 45}} )^T.
\end{eqnarray*}
The number of observations is set to be $n=200$.  For each model, we generate $T=10000$ datasets of predictors $\bm{x}_i$ ($i=1,\dots,n$) according to $N(\bm{0}, (1-\rho)\bm{I} + \rho \bm{1}\bm{1}^T)$.  Here, we consider the case of $\rho=0.0$ and $\rho = 0.6$.  Responses $y_i$ are generated from the mixture distribution
\begin{eqnarray*}
(1-\delta) f(y|\bm{x}_i;\bm{\beta}_0) + \delta q(y) \ \  (i=1,\dots,n),	
\end{eqnarray*}
where $f(y|\bm{x};\bm{\beta}_0)$ is a density function corresponding to the LPRE defined as (\ref{hm:LPRE_pdf}), $q(y)$ is a density function of distribution of outliers, and $\delta$  ($0 \leq \delta < 1$) is an outlier ratio.  The outlier ratio is set to be $\delta =0,$ $0.05,$ $0.1,$ and $0.2$ in this simulation.  We assume that $q(y)$ follows a log-normal distribution (pdf: $q(y) = 1/(\sqrt{2\pi}y\sigma)\exp\{ -(\log y - \mu)^2 /(2\sigma^2)\}$) with $(\mu, \sigma) = (\pm 5,1)$.  When $\mu =5$, the outliers take extremely large values.  On the other hand, when $\mu =-5$, the data values of outliers are nearly zero.  

\subsection{Investigation of relative prediction error and mean squared error of the estimator}
To investigate the performance of our proposed procedure, we use the relative prediction error (RPE) and the mean square error (MSE) for the $t$th dataset, defined as
\begin{eqnarray}
RPE(t) &=& \sum_{i=1}^n \frac{[y_i^{\rm new}(t) - \exp\{\bm{x}_i(t)^T\hat{\bm{\beta}}(t)\} ]^2}{y_i^{\rm new}(t)\exp\{\bm{x}_i(t)^T\hat{\bm{\beta}}(t)\}}, \label{eq:RPE} \\
MSE(t) &=& \|\hat{\bm{\beta}}(t) - \bm{\beta}_0  \|^2, \label{eq:MSE}
\end{eqnarray}
respectively, 
where $\hat{\bm{\beta}}(t)$ is an estimator obtained from the dataset $\{(\bm{x}_i(t),y_i(t)); \ i=1,\dots,n\}$, and $y_i^{\rm new}(t)$ is an observation from $y_i^{\rm new}(t)|\bm{x}_i(t)$.  Here, $y_i^{\rm new}(t)|\bm{x}_i(t)$ follows a distribution of $f(y|\bm{x}_i(t);\bm{\beta}_0)$ and is independent of $y_i(t)|\bm{x}_i(t)$. Figure \ref{fig:error} shows the median and error bar of \{RPE(1), $\dots$, RPE($T$)\} and \{MSE(1), $\dots$, MSE($T$)\}.  The error bars are delineated by the  25th and 75th percentiles.   
\begin{figure}[H]
\centering
\includegraphics[width=14 cm]{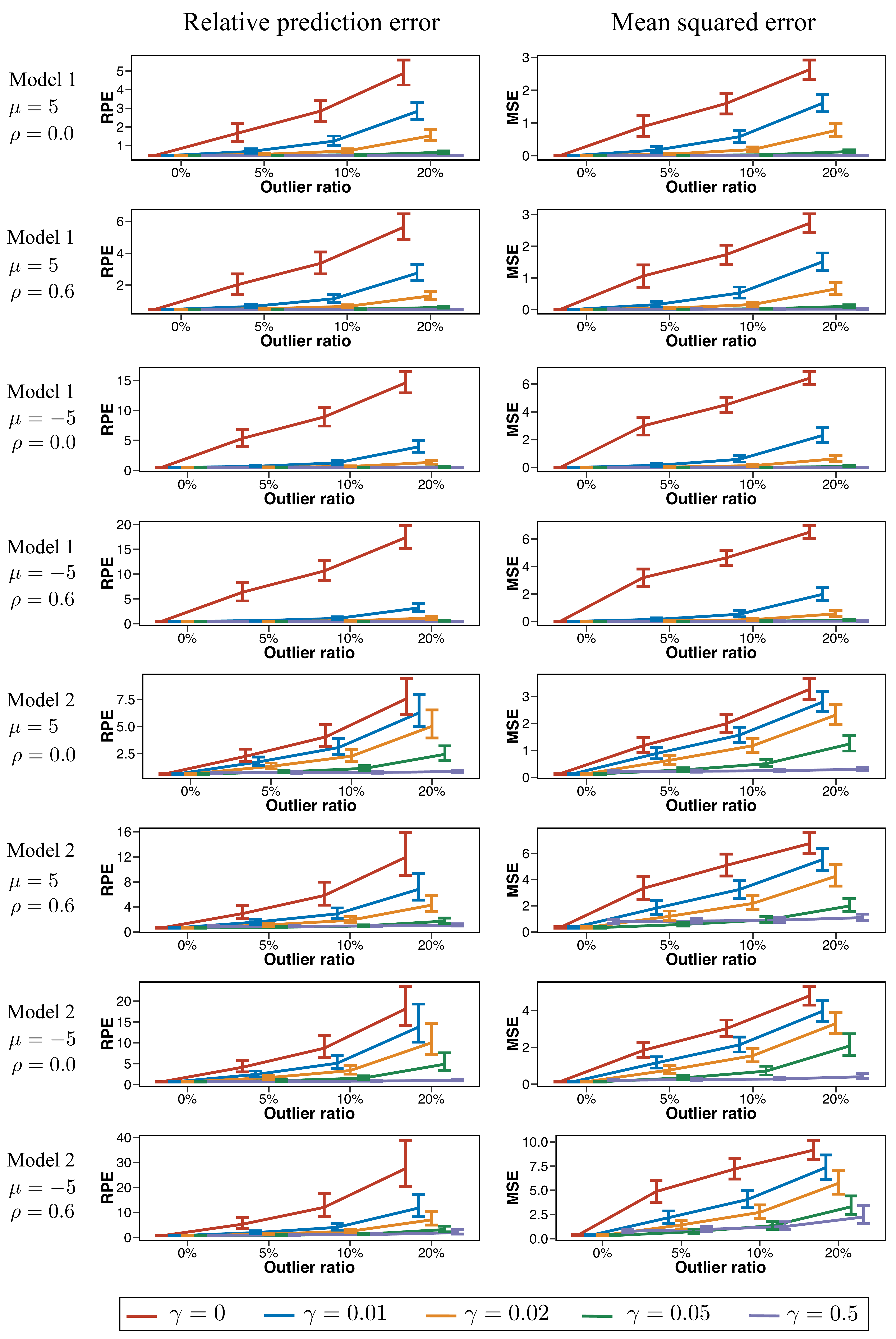}
\caption{Median and error bar of relative prediction error (RPE) in (\ref{eq:RPE}) and mean squared error (MSE) of $\bm{\beta}$ in (\ref{eq:MSE}) when parameters of the log-normal distribution (distribution of outliers) are $(\mu,\sigma)=(\pm5,1)$.  The error bars are delineated  by 25th and 75th percentiles.}
\label{fig:error}
\end{figure}   
We observe the following tendencies from the results in Figure \ref{fig:error}:
\begin{itemize}
	\item As the outlier ratio increases, the performance becomes worse in all cases.  Interestingly, the length of the error bar of RPE increases as the outlier ratio increases.   
	\item The proposed method becomes robust against outliers as the value of $\gamma$ increases.  We observe that a too large $\gamma$, such as $\gamma=10$, leads to extremely poor RPE and MSE, because most observations are regarded as outliers.  Therefore, the not too large $\gamma$, such as the $\gam=0.5$ used here, generally results in better estimation accuracy than the MLE.   
	\item The cases of $\rho = 0.6$, where the predictors are correlated, are worse than those of $\rho=0$.  Particularly, when $\gamma=0$, the value of RPE of $\rho=0.6$ becomes large on the large outlier ratio.   Still, increasing $\gam$ has led to better estimation performance uniformly.
	\item The results for different simulation models on the same value of $\gamma$ are generally different, which implies the appropriate value of $\gamma$ may change according to the data generating mechanisms.  
\end{itemize}

\subsection{Investigation of asymptotic distribution}
The asymptotic distribution is derived under the assumption that the true distribution of $y|\bm{x}_i$ follows $f(y|\bm{x}_i;\bm{\beta}_0)$, that is, $\delta = 0$.  However, we expect that, when $\gamma$ is sufficiently large and $\del$ is moderate, the asymptotic distribution may approximate the true distribution well, a point underlined by  \citet[Theorem 5.1]{Fujisawa:2008dq} in the case of i.i.d. data. We investigate whether the asymptotic distribution given by \eqref{hm:thm1-an}  appropriately works when there exist outliers.   

Let $\bm{z} = (z_1,\dots,z_p)^T$ be
\begin{equation}
\bm{z} := \sqrt{n} \left\{ {\rm diag} \left( J_{\gam}^{-1}\Delta_{\gam} J_{\gam}^{-1} \right)\right\}^{-\frac{1}{2}}\left(\hat{\bm{\beta}}_{\gam}-\bm{\beta}_0 \right).
\nn
\end{equation}
\eqref{hm:thm1-an} implies that   
\begin{equation}
z_j  \cil N\left(0,\,1\right), \quad (j=1,\dots,p).
\nn
\end{equation}
We expect that the histogram of $z_j$ obtained by the simulation would approximate the density function of the standard normal distribution when there are no (or a few) outliers.  When there exists a number of outliers, the asymptotic distribution of $z_j$ may not be $N(0,1)$ but is expected to be close to $N(0,1)$ for large $\gamma$.    Figure \ref{fig:histmu5} shows the histograms of $T=10000$ samples of $z_2$ along with the density function of the standard normal distribution for  $\mu=5$ in Model 1. 
\begin{figure}[H]
\centering
\includegraphics[width=14 cm]{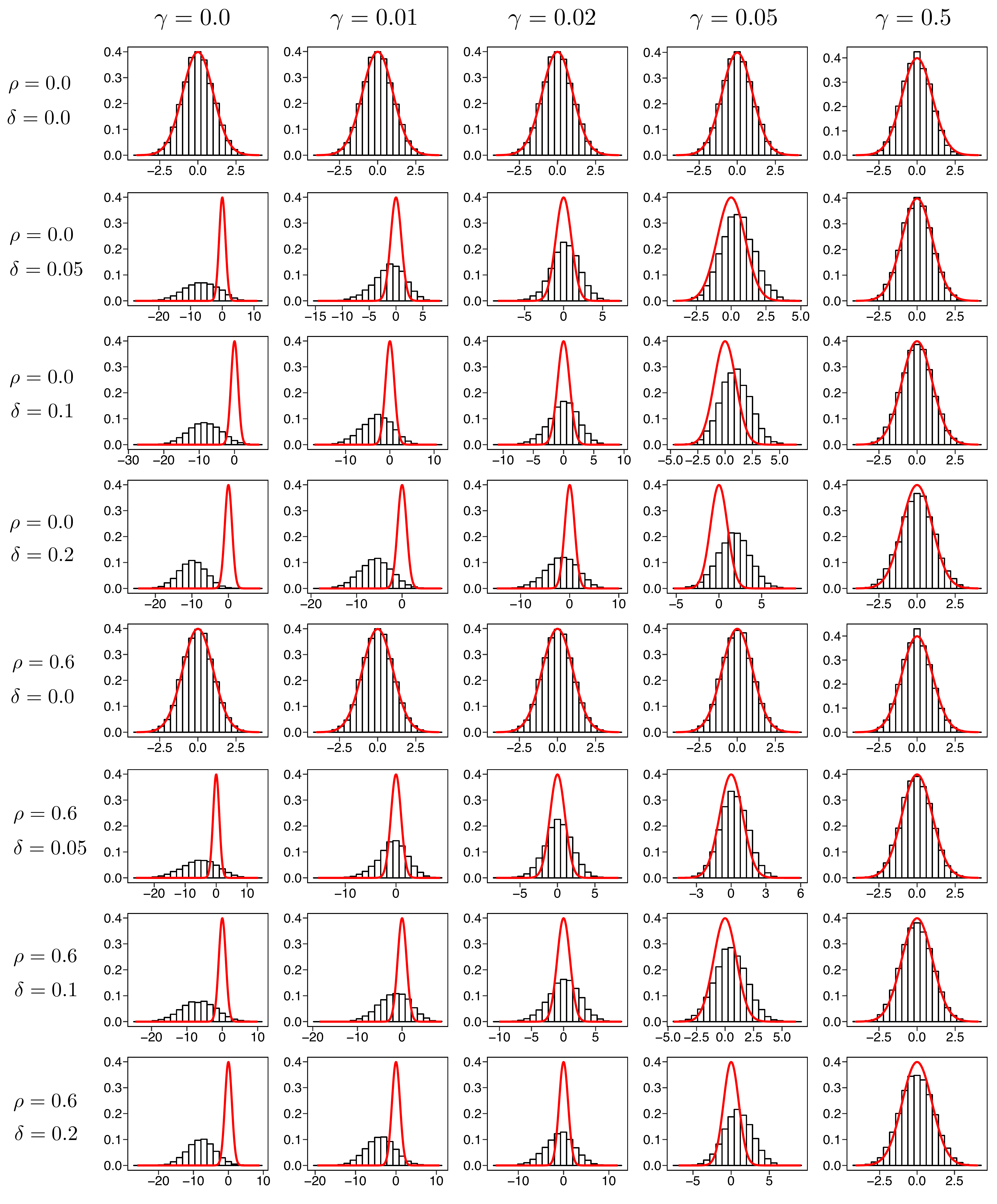}
\caption{Histograms of $T=10000$ samples of $z_2$ along with the density function of standard normal distribution for  $\mu=5$ in Model 1. }
\label{fig:histmu5}
\end{figure}   
When there are no outliers, the distribution of $z_2$ is close to the standard normal distribution whatever the value of $\gamma$ is selected.  When the outlier ratio is large, the histogram of $z_2$ is far from the density function of $N(0,1)$ for a small $\gamma$.  However, when the value of $\gamma$ is large, the histogram of $z_2$ is close to the density function of $N(0,1)$, which implies the asymptotic distribution in \eqref{hm:thm1-an} appropriately approximates the distribution of estimators even when there exist outliers.   We observe that the result of the asymptotic distributions for other $z_j$s shows a similar tendency to that of $z_2$.

\section{Real data analysis}\label{sec:real data analysis}
We apply the proposed method to electricity consumption data from the UCI Machine Learning repository \citep{Dua:2017}, which is available to download at \url{https://archive.ics.uci.edu/ml/datasets/ElectricityLoadDiagrams20112014}.  The dataset consists of 370 household electricity consumption observations from January 2011 to December 2014.  The electricity consumption is in kWh at 15 minutes intervals.  We consider the problem of prediction of the electricity consumption for next day by using past electricity consumption.  The prediction of the day ahead electricity consumption is needed when we trade electricity on markets, such as the European Power Exchange (EPEX) day ahead market (\url{https://www.epexspot.com/en/market-data/dayaheadauction}). 

To investigate the effectiveness of the proposed procedure, we choose one household that includes small positive values of electricity consumption.  The consumption data for December 25, 2014 were deleted, because the corresponding data values are zero.  We predict the electricity consumption from January 2012 to December 2014 (the data in 2011 are used only for estimating the parameter).  The actual electricity consumption data from January 2012 to December 2014 are depicted in Figure \ref{fig:rawdata}.
\begin{figure}[H]
\centering
\includegraphics[width=14 cm]{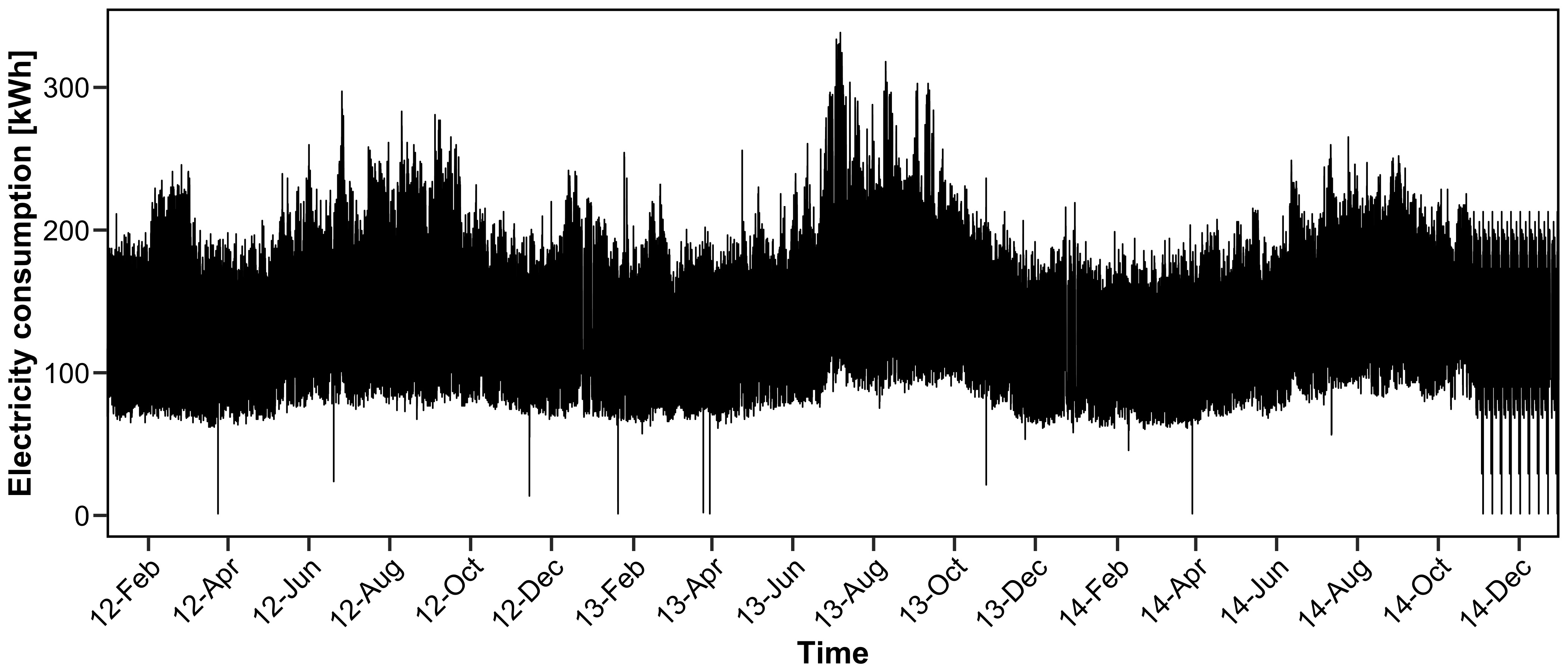}
\caption{Electricity consumption from January 2012 to December 2014 for one of the 370 households.}
\label{fig:rawdata}
\end{figure}   
Several data values are close to zero.  Particularly, from October to December 2014, exist several spikes that attain nearly zero values.  In this case, the estimation accuracy is poor with ordinary GRE criteria, as shown in our numerical simulation in the previous section.  

We assume the multiplicative regression model in (\ref{eq:multiplicative regression model}) to predict electricity consumption.  Let $y_t$ denote the electricity consumption at $t$ ($t=1,\dots,T$).  The number of observations is $T = (365\times 3 + 366 -1)\times 96 = 146,160 $.  Here, 96 is the number of measurements in one day, because electricity demand is expressed in 15 minutes intervals.  We define $\bm{x}_t$ as $\bm{x}_t = (y_{t-d},\dots,y_{t-dq})^T$, where $d=96$.  In our model, the electricity consumption at $t$ is explained by the electricity consumption of the past $q$ days for the same period.  We set $q=5$ for data analysis and use past $n=100$ days of observations to estimate the model.  

The model parameters are estimated by robust LPRE.  The values of $\gamma$ are set to be regular sequences from 0 to 0.1, with increments of 0.01.  To minimize the negative $\gamma$-likelihood function, we apply our proposed MM algorithm.  As the electricity consumption pattern on weekdays is known to be completely different from that on weekends, we make predictions for weekdays and weekends separately.  The results of the relative prediction error are depicted in Figure \ref{fig:loss_realdata}.
\begin{figure}[H]
\centering
\includegraphics[width=14 cm]{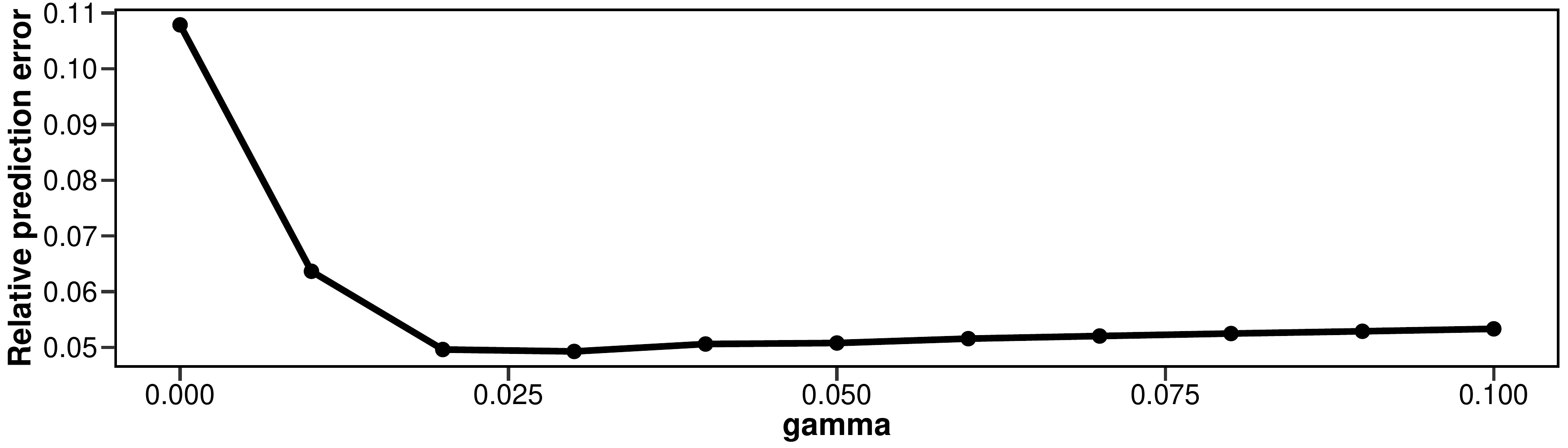}
\caption{Relative prediction error for various values of $\gamma$ for household electricity consumption data.}
\label{fig:loss_realdata}
\end{figure}
  The relative prediction error is large when $\gamma=0$ (i.e., ordinary LPRE estimation).  The minimum value of relative prediction error is 0.049 and the corresponding value of $\gamma$ is $\gamma=0.03$.  When we set a too large value of $\gamma$, efficiency decreases and the relative prediction error might increase.   

\begin{figure}[H]
\centering
\includegraphics[width=14 cm]{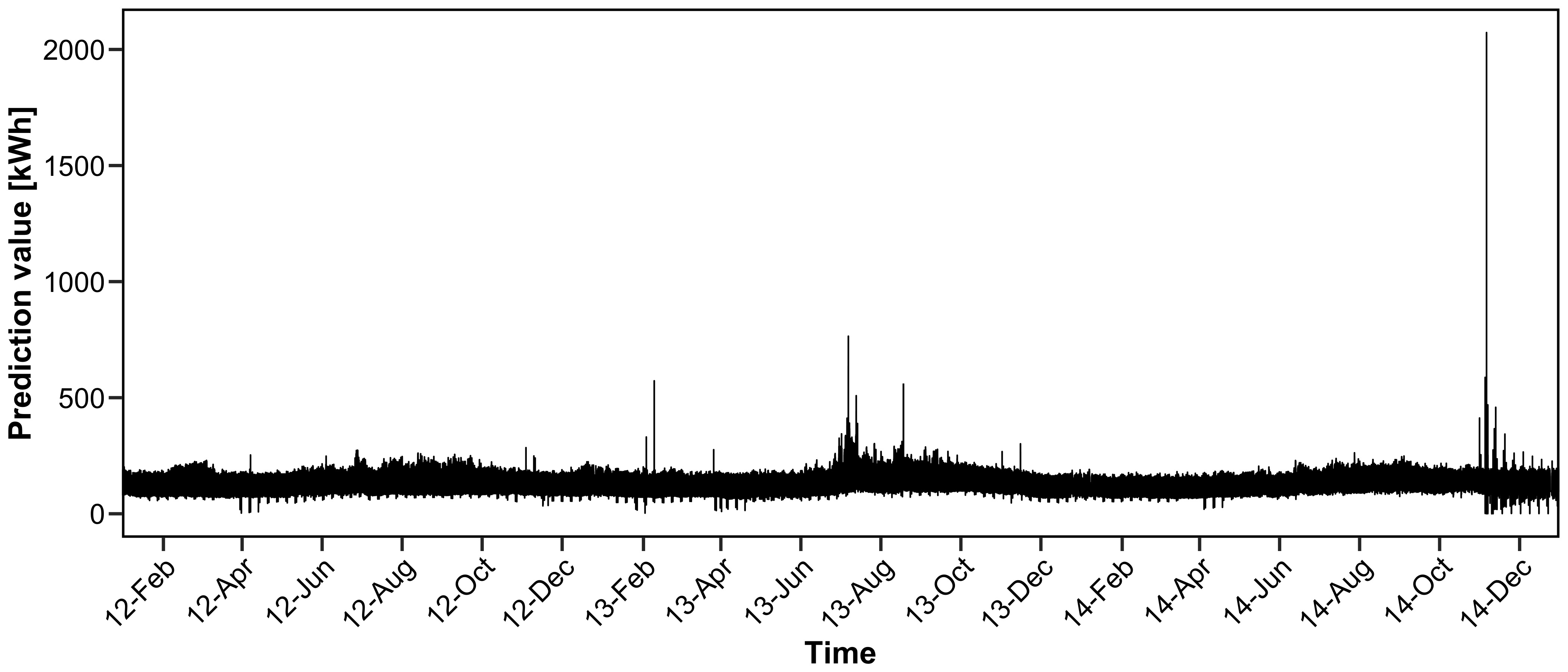}
\caption{Prediction value based on LPRE loss for household electricity consumption data.}
\label{fig:prediction_gam0}
\end{figure}   
Figure \ref{fig:prediction_gam0} shows the prediction value when $\gamma=0$.  We observe there exist several extremely large prediction values (e.g., July 8, 2013 and November 6, 2014) due to the model parameters, which are heavily affected by the nearly zero values of electricity consumption.   

\begin{figure}[H]
\centering
\includegraphics[width=14 cm]{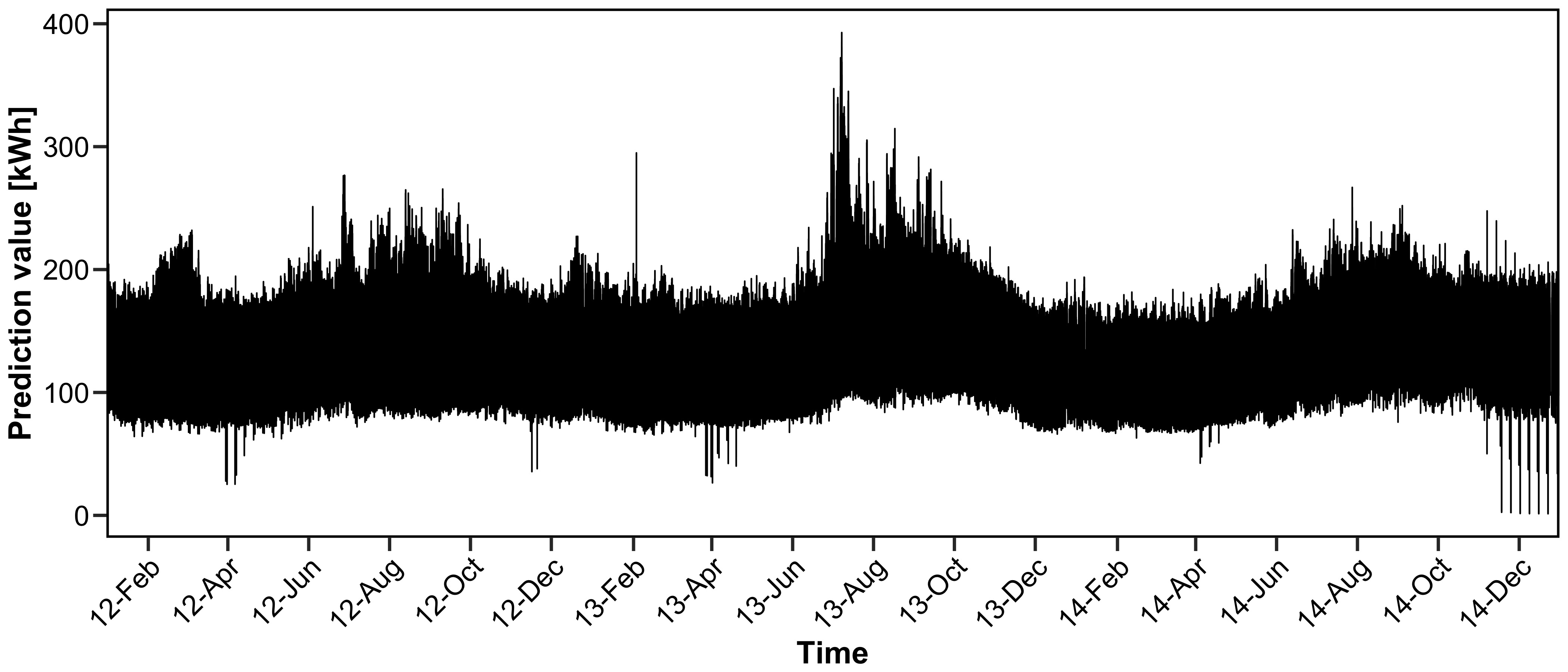}
\caption{Prediction value based on the proposed method with $\gamma=0.03$ for household electricity consumption data.}
\label{fig:prediction_gam003}
\end{figure}   
Figure \ref{fig:prediction_gam003} shows the prediction values when $\gamma=0.03$.  Extremely large prediction values are not observed and the prediction values are similar to the actual electricity demand in Figure \ref{fig:rawdata}.  Therefore, our proposed procedure is robust against outliers.   

Furthermore, to investigate the stableness of the MM algorithm described in section \ref{sec:algorithm}, we also apply the BFGS method to obtain the minimizer of the negative $\gamma$-likelihood function.  The {\tt optim} function in {\tt R} is used to implement the BFGS method.  With the BFGS method, relative prediction errors diverge when $\gamma \geq 0.03$.  Consequently, the MM algorithm is more stable than the BFGS algorithm for this dataset.

\section{Discussion}\label{sec:discussion}
We proposed a relative error estimation procedure that is robust against outliers.  The proposed procedure is based on the $\gamma$-likelihood function, which is constructed by $\gamma$-cross entropy \citep{Fujisawa:2008dq}.  Our simulation results showed the proposed method performed better than the ordinary relative error estimation procedures in terms of prediction accuracy.  Furthermore, the asymptotic distribution of the estimator yielded a good approximation, with an appropriate value of $\gamma$, even when there existed outliers.  The proposed method was applied to electricity consumption data, which included small positive values.  Although the ordinary LPRE was sensitive to small positive values, our method was able to appropriately eliminate the negative effect of these values.

As shown in Remark \ref{remark:lasso}, our method may be extended to $L_1$ regularization.  An important point in the regularization procedure is the selection of a regularization parameter. \citet{Hao:2016ei} suggested using the BIC-type criterion of Wang {\it et al}., \citep{Wang:2007fi, Wang:2009ja} for the ordinary LPRE estimator.  As a future research topic, it would be interesting to consider the problem of regularization parameter selection in high-dimensional robust relative error estimation.

\section*{Acknowledgments}
This work was partially supported by the Japan Society for the Promotion of Science KAKENHI 15K15949, and the Center of Innovation Program (COI) from JST, Japan (K. Hirose), and JST CREST Grant Number JPMJCR14D7 (H. Masuda).


\appendix

\section{Proof of Theorem \ref{hm:thm1}}
\label{hm:sec_proof}

All the asymptotic will be taken under $n\to\infty$.
We write $a_n\lesssim b_n$ if there exists a positive universal constant such that $a_n\le cb_n$ for every $n$ large enough.
For any random functions, $X_n$ and $X_0$ on $\overline{\mcb}$ we denote $X_n(\bbeta)\ucip X_0(\bbeta)$ if $\sup_{\bbeta\in\overline{\mcb}}|X_n(\bbeta)-X_0(\bbeta)| \cip 0$; below, we will simply write $\sup_{\bbeta}$ for $\sup_{\bbeta\in\overline{\mcb}}$.

First, we state a preliminary lemma, which will be repeatedly used in the sequel.

\begin{lemma}
\label{hm:lem-1}
Let $\eta(\bm{x};\bm{\beta})$ and $\zeta(\bm{x},y;\bm{\beta})$ be vector-valued measurable functions satisfying that
\begin{align}
\sup_{\bbeta}\max_{k\in\{0,1\}}\left|\p_{\bbeta}^{k}\eta(\bm{x};\bm{\beta})\right| &\le \overline{\eta}(\bm{x}), \nn\\
\sup_{\bbeta}\max_{k\in\{0,1\}}\left|\p_{\bbeta}^{k}\zeta(\bm{x},y;\bm{\beta})\right| &\le \overline{\zeta}(\bm{x},y)
\nonumber
\end{align}
for some $\overline{\eta}$ and $\overline{\zeta}$ such that
\begin{equation}
\overline{\eta} + \int_0^\infty \overline{\zeta}(\cdot,y) dy \in \bigcap_{q>0}L^{q}(\pi)
\nonumber
\end{equation}
for every $q>0$. Then, 
\begin{align}
\frac{1}{n}\sumi \eta(\bm{x}_i;\bm{\beta}) &\ucip \int_{\mcx} \eta(\bm{x};\bm{\beta}) \pi(d\bm{x}), \label{hm:lem-1-1}\\
\frac{1}{n}\sumi \zeta(\bm{x}_i,y_i;\bm{\beta}) &\ucip \int_0^\infty \int_{\mcx} \zeta(\bm{x},y;\bm{\beta}) f(y|\bm{x},\bbeta_0)\pi(d\bm{x})dy.
\label{hm:lem-1-2}
\end{align}
\end{lemma}

\begin{proof}
\eqref{hm:lem-1-1} is a special case of \eqref{hm:lem-1-2}, hence we only show the latter.
Observe that
\begin{align}
& \sup_{\bbeta}\left| \frac{1}{n}\sumi \zeta(\bm{x}_i,y_i;\bm{\beta}) - \iint \zeta(\bm{x},y;\bm{\beta}) f(y|\bm{x},\bbeta_0)\pi(d\bm{x})dy \right| \nn\\
&\le \frac{1}{\sqrt{n}} \sup_{\bbeta}\left| \sumi \frac{1}{\sqrt{n}} \left( \zeta(\bm{x}_i,y_i;\bm{\beta}) - \int \zeta(\bm{x}_i,y;\bm{\beta}) f(y|\bm{x}_i,\bbeta_0)dy\right) \right| \nn\\
&{}\qquad 
+ \sup_{\bbeta}\left| \frac{1}{n}\sumi \int \zeta(\bm{x}_i,y;\bm{\beta}) f(y|\bm{x}_i,\bbeta_0)dy - \iint \zeta(\bm{x},y;\bm{\beta}) f(y|\bm{x},\bbeta_0)\pi(d\bm{x})dy \right| \nn\\
&=: \frac{1}{\sqrt{n}} \sup_{\bbeta}\left| M_n(\bm{\beta}) \right| + \sup_{\bbeta}\left| C_n(\bm{\beta}) \right|.
\nn
\end{align}
For the first term, let us recall the Sobolev inequality \citep[Section 10.2]{Fri06}:
\begin{equation}
E\left(\sup_{\bbeta}\left| M_n(\bbeta)\right|^{q}\right)
\lesssim \sup_{\bbeta}E\left\{|M_n(\bbeta)|^{q}\right\} + \sup_{\bbeta}E\left\{|\p_{\bbeta}M_n(\bbeta)|^{q}\right\}
\label{hm:sobolev.ineq}
\end{equation}
for $q>p$. The summands of $M_n(\bbeta)$ trivially form a martingale difference array with respect to the filtration $\mcf_j := \sig(\bm{x}_i;\, i\le j)$, $j\in\mbbn$:
since we are assuming that the conditional distribution of $y_i$ given $\{(\bm{x}_i,\bm{x}_{i-1},\bm{x}_{i-2},\dots),\,(y_{i-1},y_{i-2},\dots)\}$ equals that given $\bm{x}_i$ (Section \ref{hk:sec_gam-div}), each summand of $M_n(\bbeta)$ equals $\frac{1}{\sqrt{n}} \left( \zeta(\bm{x}_i,y_i;\bm{\beta}) - E\{\zeta(\bm{x}_i,y;\bm{\beta})|\,\mcf_{j-1}\}\right)$.
Hence, by means of the Burkholder's inequality for martingales, we obtain, for $q>p\vee 2$,
\begin{align}
\sup_{\bbeta}E\left\{|M_n(\bbeta)|^{q}\right\} &\lesssim 
\sup_{\bbeta}\frac{1}{n} \sumi E\left\{\left|
\left( \zeta(\bm{x}_i,y_i;\bm{\beta}) - \int \zeta(\bm{x}_i,y;\bm{\beta}) f(y|\bm{x}_i,\bbeta_0)dy\right)
\right|^{q}\right\} < \infty.
\nonumber
\end{align}
We can take the same route for the summands of $\p_{\bbeta}M_n(\bbeta)$ to conclude that $\sup_{\bbeta}E\left\{|\p_{\bbeta}M_n(\bbeta)|^{q}\right\}<\infty$.
These estimates, combined with \eqref{hm:sobolev.ineq}, lead to the conclusion that
\begin{equation}
\sup_{\beta}\left| M_n(\bm{\beta}) \right|=O_p(1).
\nonumber
\end{equation}
As for the other term, we have $C_n(\bm{\beta}) \cip 0$ for each $\bm{\beta}$ and also
\begin{equation}
\sup_n E\left( \sup_{\bm{\beta}}\left|\p_{\bm{\beta}}C_n(\bm{\beta})\right|  \right)<\infty.
\nonumber
\end{equation}
The latter implies the tightness of the family $\{C_n(\bm{\beta})\}_n$ of continuous random functions on the compact set $\overline{\mcb}$, thereby entailing that $C_n(\bm{\beta}) \ucip 0$. The proof is complete.
\end{proof}

\subsection{Consistency}

Let $f_i(\bm{\beta}) := f(y_i| \bm{x}_{i};\bm{\beta})$ for brevity and
\begin{align}
A_{\gam,n}(\bm{\beta}) &:= \frac{1}{n}\sum_{i=1}^n f_i(\bm{\beta})^{\gamma}, \nn\\
\overline{A}_{\gam,n}(\bm{\beta}) &:= \frac{1}{n}\sum_{i=1}^n \int f(y|\bm{x}_i;\bm{\beta})^{\gamma+1}dy
=C(\gam,h) \frac{1}{n}\sum_{i=1}^n \exp(- \gam \bm{x}_i^{T}\bbeta).
\nonumber
\end{align}
By means of Lemma \ref{hm:lem-1}, we have
\begin{align}
A_{\gam,n}(\bm{\beta}) &\ucip A_\gam(\bm{\beta}) := \iint f(y|\bm{x};\bm{\beta})^{\gam}f(y|\bm{x},\bbeta_0)\pi(d\bm{x})dy, \nn\\
\overline{A}_{\gam,n}(\bm{\beta}) &\ucip \overline{A}_{\gam}(\bm{\beta}) := \iint f(y|\bm{x};\bm{\beta})^{\gam+1}\pi(d\bm{x})dy
=C(\gam,h) \int \exp(- \gam \bm{x}^{T}\bbeta)\pi(d\bm{x}).
\nonumber
\end{align}
Since $\inf_{\bbeta}\left\{A_\gam(\bm{\beta})\wedge \overline{A}_{\gam}(\bm{\beta})\right\} >0$, we see that taking the logarithm preserves the uniformity of the convergence in probability:
for the $\gam$-likelihood function \eqref{gamma-div}, it holds that
\begin{equation}
\ell_{\gamma,n}(\bm{\beta}) \ucip \ell_{\gam,0}(\bm{\beta}) := -\frac{1}{\gamma} \log \left\{ A_\gam(\bm{\beta}) \right\} 
+ \frac{1}{1+\gamma} \log \left\{ \overline{A}_{\gam}(\bm{\beta}) \right\}.
\label{hm:consistency.proof-1}
\end{equation}
The limit equals the $\gam$-cross entropy from $g(\cdot|\cdot)=f(\cdot|\cdot;\bbeta_0)$ to $f(\cdot|\cdot;\bbeta)$.
We have $\ell_{\gam,0}(\bm{\beta})\ge \ell_{\gam,0}(\bm{\beta}_0)$, the equality holding if and only if $f(\cdot|\cdot;\bbeta_0)=f(\cdot|\cdot;\bbeta)$ (see \citep[Theorem 1]{Kawashima:2017ho}).
By \eqref{hm:add.eq-1} the latter condition is equivalent to $\rho( e^{-\bm{x}^{T}\bbeta} y) = \rho( e^{-\bm{x}^{T}\bbeta_0} y)$, followed by $\bbeta=\bbeta_0$ from Assumption \ref{hm:A-iden}.
This, combined with \eqref{hm:consistency.proof-1} and the argmin theorem (cf. \citep[Chapter 5]{vdV98}), concludes the consistency $\hat{\bbeta}_\gam \cip \bbeta_0$.
[Note that we do not need Assumption \ref{hm:A-iden} if $\ell_{\gamma,n}$ is a.s. convex, which generally may not be the case for $\gam>0$.]


\subsection{Asymptotic normality}

First we note that Assumption \ref{hm:A-ve} ensures that, for every $\al>0$, there corresponds a function $\overline{F}_{\al} \in L^{1}\left(f(y|\bm{x},\bbeta_0)\pi(d\bm{x})dy\right)$ such that
\begin{equation}
\max_{k=0,1,2,3}\sup_{\bbeta} \left| \p_{\bbeta}^{k}\left\{f(y|\bm{x},\bbeta)^{\al}\right\}
\right| \le \overline{F}_{\al}(\bm{x},y).
\nonumber
\end{equation}
This estimate will enable us to interchange the order of $\p_{\bbeta}$ and the $dy$-Lebesgue integration, repeatedly used below without mention.

Let $s_i(\bm{\beta}) = s(y_i| \bm{x}_{i};\bm{\beta})$, and
\begin{align}
S_{\gam,n}(\bm{\beta}) &:= \frac{1}{n}\sum_{i=1}^n f_i(\bm{\beta})^{\gamma}s_i(\bm{\beta}), \nn\\
\overline{S}_{\gam,n}(\bm{\beta}) &:= \frac{1}{n}\sum_{i=1}^n \int f(y|\bm{x}_i;\bm{\beta})^{\gamma+1}s(y|\bm{x}_i;\bm{\beta})dy.
\nonumber
\end{align}
Then, the $\gam$-likelihood equation $\p_{\bbeta}\ell_{\gam,n}(\bbeta)=\bm{0}$ is equivalent to
\begin{equation}
\Psi_{\gam,n}(\bbeta):=\overline{A}_{\gam,n}(\bm{\beta})S_{\gam,n}(\bm{\beta}) - A_{\gam,n}(\bm{\beta})\overline{S}_{\gam,n}(\bm{\beta}) = \bm{0}.
\nonumber
\end{equation}
By the consistency of $\hat{\bbeta}_\gam$, we have $P(\hat{\bbeta}_\gam\in\mcb)\to 1$; hence $P\{\Psi_{\gam,n}(\hat{\bbeta}_\gam)=\bm{0}\}\to 1$ as well. Recall that $\mcb$ is open.
Therefore, virtually defining $\hat{\bbeta}_\gam$ to be $\bbeta_0\in\mcb$ if $\Psi_{\gam,n}(\hat{\bbeta}_\gam)=\bm{0}$ has no root, we may and do proceed as if $\Psi_{\gam,n}(\hat{\bbeta}_\gam)=\bm{0}$ a.s.
Because of the Taylor expansion
\begin{equation}
\left( -\int_0^1 \p_{\bbeta}\Psi_{\gam,n}\left(\bbeta_0 + s(\hat{\bbeta}_n-\bbeta_0)\right)ds  \right)
\sqrt{n}\left(\hat{\bm{\beta}}_{\gam}-\bm{\beta}_0 \right) = \sqrt{n}\Psi_{\gam,n}(\bbeta_0),
\nonumber
\end{equation}
to conclude \eqref{hm:thm1-an} it suffices to show that (recall the definitions \eqref{hm:D_def} and \eqref{hm:J_def})
\begin{align}
\sqrt{n}\Psi_{\gam,n}(\bbeta_0) &\cil N_p\left(\bm{0},\, \D_\gam \right), \label{hm:an.proof-1} \\
-\p_{\bbeta}\Psi_{\gam,n}(\hat{\bbeta}'_n) &\cip J_\gam \quad\text{for every $\hat{\bbeta}'_n \cip \bbeta_0$.}
\label{hm:an.proof-2}
\end{align}

First we prove \eqref{hm:an.proof-1}.
By direct computations and Lemma \ref{hm:lem-1}, we see that
\begin{align}
\sqrt{n}\Psi_{\gam,n} &= 
\overline{A}_{\gam,n} \sqrt{n}\left(S_{\gam,n}-\overline{S}_{\gam,n}\right) 
- \sqrt{n}\left( A_{\gam,n} - \overline{A}_{\gam,n}\right)\overline{S}_{\gam,n} \nn\\
&= \sumi \frac{1}{\sqrt{n}} \left\{ H'_\gam \left( f_i^{\gamma}s_i - \int f(y|\bm{x}_i;\bm{\beta}_0)^{\gamma+1}s(y|\bm{x}_i;\bm{\beta}_0)dy\right) - \left( f_i^{\gamma} - \int f(y|\bm{x}_i;\bm{\beta}_0)^{\gamma+1}dy\right) H''_\gam \right\} \nn\\
&=:\sumi \chi_{\gam,i}.
\nonumber
\end{align}
The sequence $(\chi_{\gam,i})_{i\le n}$ is an $(\mcf_j)$-martingale-difference array.
It is easy to verify the Lapunov condition:
\begin{equation}
\exists \al>0,\quad \sup_{n}\sup_{i\le n} E\left(\left|\chi_{\gam,i}\right|^{2+\al}\right)<\infty.
\nonumber
\end{equation}
Hence the martingale central limit theorem concludes \eqref{hm:an.proof-1} if we show the following convergence of the quadratic characteristic:
\begin{equation}
\frac{1}{n}\sumi E\Big(\chi_{\gam,i}^{\otimes 2}\Big|\mcf_{i-1}\Big) \cip \D_\gam.
\nonumber
\end{equation}
This follows on observing that
\begin{align}
& \frac{1}{n}\sumi E\Big(\chi_{\gam,i}^{\otimes 2}\Big|\mcf_{i-1}\Big) \nn\\
&= (H'_\gam)^2 \frac{1}{n}\sumi \var\big( f_j^{\gam}s_j \big| \mcf_{i-1}\big)
 + (H''_\gam)^{\otimes 2} \frac{1}{n}\sumi \var\big( f_j^{\gam}\big| \mcf_{i-1}\big)
 -2 H'_\gam \frac{1}{n}\sumi \cov\big( f_j^{\gam}s_j, f_j^{\gam}\big| \mcf_{i-1}\big) H''_\gam \nn\\
& = (H'_\gam)^2 \left\{ \iint f(y|\bm{x};\bbeta_0)^{2\gam+1}s(y|\bm{x};\bbeta_0)^{\otimes 2}dy\pi(d\bm{x}) -(H''_\gam)^{\otimes 2}\right\} \nn\\
&{}\qquad 
 + (H''_\gam)^{\otimes 2} \left\{ \iint f(y|\bm{x};\bbeta_0)^{2\gam+1}dy\pi(d\bm{x}) - (H'_\gam)^{2}\right\} \nn\\
&{}\qquad 
 -2 H'_\gam \left\{ \iint f(y|\bm{x};\bbeta_0)^{2\gam+1}s(y|\bm{x};\bbeta_0)dy\pi(d\bm{x}) - H'_\gam H''_\gam \right\} (H''_\gam)^{T}
+o_p(1) \nn\\
&= (H'_\gam)^2 \iint f(y|\bm{x};\bbeta_0)^{2\gam+1}s(y|\bm{x};\bbeta_0)^{\otimes 2}dy\pi(d\bm{x}) \nn\\
&{}\qquad + (H''_\gam)^{\otimes 2} \iint f(y|\bm{x};\bbeta_0)^{2\gam+1}dy\pi(d\bm{x})  -2 H'_\gam \iint f(y|\bm{x};\bbeta_0)^{2\gam+1}s(y|\bm{x};\bbeta_0)dy\pi(d\bm{x}) (H''_\gam)^{T} + o_p(1)\nn\\
&=\D_\gam + o_p(1);
\nonumber
\end{align}
invoke the expression \eqref{hm:add.eq-1} for the last equality.

Next, we show \eqref{hm:an.proof-2}.
Under the present regularity condition, we can deduce that
\begin{equation}
\sup_{\bbeta}| \p_{\bbeta}^2\Psi_{\gam,n}(\bbeta) | = O_p(1).
\nonumber
\end{equation}
It therefore suffices to verify that $-\p_{\bbeta}\Psi_{\gam,n}(\bbeta_0) \cip J_\gam(\bbeta_0)=J_\gam$.
 This follows from a direct computation of $-\p_{\bbeta}\Psi_{\gam,n}(\bbeta_0)$, combined with the applications of Lemma \ref{hm:lem-1}:
\begin{align}
-\p_{\bbeta}\Psi_{\gam,n}(\bbeta_0)
&= A_{\gam,n}\left(\frac{1}{n}\sum_{i=1}^n \int f(y|\bm{x}_i;\bbeta_0)^{\gamma+1}s(y|\bm{x}_i;\bbeta_0)^{\otimes 2}dy\right) 
- \overline{S}_{\gam,n} S_{\gam,n}^{T}
\nn\\
&{}\qquad + \gam\left( S_{\gam,n}\overline{S}_{\gam,n}^{T} - \overline{S}_{\gam,n} S_{\gam,n}^{T} \right) \nn\\
&{}\qquad + \gam\left\{
A_{\gam,n} \left(\frac{1}{n}\sum_{i=1}^n \int f(y|\bm{x}_i;\bbeta_0)^{\gamma+1}s(y|\bm{x}_i;\bbeta_0)^{\otimes 2}dy\right) 
- \overline{A}_{\gam,n} \left(\frac{1}{n}\sum_{i=1}^n  f_i^{\gamma}s_i^{\otimes 2}\right) 
\right\} \nn\\
&{}\qquad + \left\{
A_{\gam,n} \left(\frac{1}{n}\sum_{i=1}^n \int f(y|\bm{x}_i;\bbeta_0)^{\gamma+1}\p_{\bbeta}s(y|\bm{x}_i;\bbeta_0)dy \right) 
-\overline{A}_{\gam,n} \left(\frac{1}{n}\sum_{i=1}^n f_i^{\gamma}\p_{\bbeta}s_i \right)
\right\}
\nn\\
&=
A_{\gam,n}\left(\frac{1}{n}\sum_{i=1}^n \int f(y|\bm{x}_i;\bbeta_0)^{\gamma+1}s(y|\bm{x}_i;\bbeta_0)^{\otimes 2}dy\right) 
- \overline{S}_{\gam,n} S_{\gam,n}^{T} + o_p(1) \nn\\
&= H'_\gam \iint f(y|\bm{x};\bbeta_0)^{\gam+1}s(y|\bm{x};\bbeta_0)^{\otimes 2}dy\pi(d\bm{x}) - (H''_\gam)^{\otimes 2} + o_p(1) \nn\\
&= J_\gam + o_p(1).
\nonumber
\end{align}

\subsection{Consistent estimator of the asymptotic covariance matrix}

Thanks to the stability assumptions on the sequence $\bm{x}_1, \bm{x}_2,\dots$, we have
\begin{equation}
\frac1n \sumi \bm{x}_i^{\otimes k}\exp(- \gam \bm{x}_i^{T}\bbeta_0) \cip \Pi_{k}(\gam), \qquad k=0,1,2.
\nonumber
\end{equation}
Moreover, we can find some constants $\del, \del'>0$ for which
\begin{align}
& \left| \frac1n \sumi \bm{x}_i^{\otimes k}\exp(- \gam \bm{x}_i^{T}\bbeta_0) 
- \frac1n \sumi \bm{x}_i^{\otimes k}\exp(- \gam \bm{x}_i^{T}\hat{\bbeta}_{\gam}) \right| \nn\\
&
\lesssim \left(\frac1n \sumi |\bm{x}_i|^{k+1}\exp\left(\del' |\bm{x}_i|^{1+\del}\right)\right)
\left| \hat{\bbeta}_{\gam}-\bbeta_0\right| = O_p(1)\left| \hat{\bbeta}_{\gam}-\bbeta_0\right| \cip 0.
\nonumber
\end{align}
These observations are enough to conclude \eqref{hm:thm1-av.ce}.

\bibliographystyle{abbrvnat}

\bibliography{paper-ref}

\end{document}